\documentclass[a4paper,11pt]{article}
\pdfoutput=1 
\usepackage[margin=1in]{geometry}
\usepackage[sc]{mathpazo}
\usepackage[utf8]{inputenc}

\usepackage{amsfonts}
\usepackage{amsmath}
\usepackage{amssymb}
\usepackage{amsthm}
\usepackage{float}

\usepackage{microtype}
\usepackage[svgnames]{xcolor}

\usepackage{xspace}
\usepackage{mdframed}
\usepackage{tikz}
\usetikzlibrary{arrows}
\usetikzlibrary{math}

\usepackage[ruled,linesnumbered]{algorithm2e} 
\usepackage{setspace}

\SetAlFnt{\small}
\SetAlCapFnt{\small}
\SetAlCapNameFnt{\small}
\SetAlCapHSkip{0pt}
\IncMargin{-\parindent}

\usepackage{authblk}

\usepackage{mathtools,bm}

\usepackage{hyperref}
\hypersetup{colorlinks={true},urlcolor={blue},linkcolor={DarkBlue},citecolor=[named]{DarkGreen},unicode}
\usepackage[authoryear,square]{natbib}
\usepackage[capitalise,nameinlink,noabbrev]{cleveref}
\usepackage{doi}

\usepackage{crossreftools}
\theoremstyle{definition}
\newtheorem{definition}{Definition}

\theoremstyle{plain}
\newtheorem{theorem}{Theorem}[section]
\newtheorem{lemma}[theorem]{Lemma}

\newtheorem{proposition}[theorem]{Proposition}

\theoremstyle{definition}
\newtheorem{remark}{Remark}

\DeclareRobustCommand{\abbrevcrefs}{%
\crefname{proposition}{Prop.}{Props.}%
\crefname{lemma}{Lem.}{Lems.}%
\crefname{theorem}{Thm.}{Thms.}%
}
\DeclareRobustCommand{\cshref}[1]{{\abbrevcrefs\cref{#1}}}

\newcommand{\e}{\ensuremath{\varepsilon}}
\newcommand{\ch}{\textsc{Consensus-Halving}}
\newcommand{\ppa}{\textup{PPA}\xspace}
\newcommand{\eval}{\textsc{eval}\xspace}
\newcommand{\cut}{\textsc{cut}\xspace}
\newcommand{\lcut}{\ensuremath{x}\xspace}
\newcommand{\rcut}{\ensuremath{y}\xspace}
\newcommand{\tucker}{\textsc{Tucker}\xspace}
\newcommand{\bu}{\textsc{Borsuk-Ulam}\xspace}
\newcommand{\calI}{\mathcal{I}}

\newcommand{\lplus}{``$+$''\xspace}
\newcommand{\lminus}{``$-$''\xspace}

\title{Two's Company, Three's a Crowd: Consensus-Halving for a Constant Number of Agents\footnote{A preliminary version of this paper appeared in the Proceedings of the 22nd ACM Conference on Economics and Computation (EC '21).}}

\author{
\textbf{Argyrios Deligkas}\smallskip

\small{Royal Holloway University of London, United Kingdom}\smallskip

\href{mailto:Argyrios.Deligkas@rhul.ac.uk}{\small{\texttt{Argyrios.Deligkas@rhul.ac.uk}}}
\and 
\textbf{Aris Filos-Ratsikas}\smallskip 

\small{University of Edinburgh, United Kingdom}\smallskip

\href{mailto:Aris.Filos-Ratsikas@ed.ac.uk}{\small{\texttt{Aris.Filos-Ratsikas@ed.ac.uk}}}
\and 
\textbf{Alexandros Hollender}\smallskip 

\small{University of Oxford, United Kingdom}\smallskip

\href{mailto:Alexandros.Hollender@cs.ox.ac.uk}{\small{\texttt{Alexandros.Hollender@cs.ox.ac.uk}}}\\
}
\date{}

\begin{document}

\maketitle

\begin{abstract}
We consider the $\e$-\ch\ problem, in which a set of heterogeneous agents aim at dividing a continuous resource into two (not necessarily contiguous) portions that all of them simultaneously consider to be of approximately the same value (up to $\e$). This problem was recently shown to be \ppa-complete, for $n$ agents and $n$ cuts, even for very simple valuation functions.
In a quest to understand the root of the complexity of the problem, we consider the setting where there is only a \emph{constant} number of agents, and we consider both the \emph{computational complexity} and the \emph{query complexity} of the problem.

For agents with \emph{monotone} valuation functions, we show a dichotomy: for two agents the problem is polynomial-time solvable, whereas for three or more agents it becomes \ppa-complete. Similarly, we show that for two monotone agents the problem can be solved with polynomially-many queries, whereas for three or more agents, we provide exponential query complexity lower bounds. These results are enabled via an interesting connection to a \emph{monotone Borsuk-Ulam} problem, which may be of independent interest. For agents with general valuations, we show that the problem is \ppa-complete and admits exponential query complexity lower bounds, even for two agents.
\end{abstract}


\section{Introduction}

The topic of fair division, founded in the work of \citet{steinhaus1949division}, has been in the centre of the literature in economics, mathematics, and more recently computer science and artificial intelligence. Classic examples include the well-known fair cake cutting problem for the division of divisible resources (e.g., see \citep{gamow1958puzzle}, \citep{brams1996fair} or \citep{procaccia2013cake}), as well as the fair division of indivisible resources (e.g., see \citep{BouveretCM16-indivisible}). 
The earlier works in economics and mathematics were mainly concerned with the questions of whether fair solutions exist, and whether they can be found constructively, i.e., via finite-time protocols. In the more recent literature in computer science, a plethora of works have been concerned with the computational complexity of finding such solutions, either by providing polynomial-time algorithms, or by proving computational hardness results. Currently, it would be no exaggeration to say that fair division is one of the most vibrant and important topics in the intersection of these areas.  

Besides the classic fair division settings mentioned above, another well-known problem is the \emph{Consensus-Halving} problem, whose origins date back to the 1940s and the work of \citet{neyman1946theoreme}. In this problem, a set of $n$ agents with different and heterogeneous valuation functions over the unit interval $[0,1]$ aim at finding a partition of the interval into pieces labelled either \lplus or \lminus using at most $n$ cuts, such that the total value of every agent for the portion labelled \lplus and for the portion labelled \lminus is the same. Very much like other well-known problems in fair division, the existence of a solution can be proven via fixed-point theorems, in particular the Borsuk-Ulam theorem \citep{Borsuk1933}, and can also be seen as a generalisation of the Hobby-Rice Theorem \citep{hobby1965moment}.

The problem has applications in the context of the well-known \emph{Necklace Splitting problem} of \citet{Alon87-Necklace} and was studied in conjunction with this latter problem \citep{Goldberg1985,Alon1986}. Other applications were highlighted by \citet{SS03-Consensus}, who were the first to study the problem in isolation. For example, consider two families that are dividing a piece of land into two regions, such that every member of each family considers the split to be equal. Another example is the 1994 Convention of the Law of the Sea (see \citep{brams1996fair}), which regards the protection of developing countries in the event that an industrialised nation is planning to mine resources in international waters. In such cases, a representative of the developing nations reserves half of the seabed for future use by them, and a consensus-halving solution would correspond to a partition of the seabed into two portions that all the developing nations consider to be of equal value. 

\citet{SS03-Consensus} in fact studied the approximate version of Consensus-Halving, coined the $\e$-\ch\ problem, where the requirement is that the total value of every agent for the portion labelled \lplus and for the portion labelled \lminus is approximately the same, up to an additive parameter $\e$. For this version, \citet{SS03-Consensus} provided a constructive solution via an exponential-time algorithm. The $\e$-\ch\ problem received considerable attention in the literature of computer science over the past few years, as it was proven to be the first ``natural'' \ppa-complete problem \citep{FRG18-Consensus}, i.e., a problem that does not have a polynomial-sized circuit explicitly in its definition, answering a decade-old open question \citep{Papadimitriou94-TFNP-subclasses, Grigni2001}. 
Additionally, \citet{FRG18-Necklace}, reduced from this problem to establish the \ppa-completeness of Necklace Splitting with two thieves; these \ppa-completeness results provided the first definitive evidence of intractability for these two classic problems, establishing for instance that solving them is at least as hard as finding a Nash equilibrium of a strategic game \citep{Daskalakis2009,chen2009settling}.
\citet{FRHSZ2020consensus-easier} improved on the results for the $\e$-\ch\ problem, by showing that the problem remains \ppa-complete, even if one restricts the attention to very small classes of agents' valuations, namely piecewise uniform valuations with only two valuation blocks. Very recently, the $\e$-\ch\ problem was shown to be \ppa-complete even for constant $\e$, namely any $\e < 1/5$ \citep{DeligkasFHM2022-CH-constant-eps}.

This latter result falls under the general umbrella of imposing restrictions on the structure of the problem, to explore if the computational hardness persists or whether we can obtain polynomial-time algorithms. \citet{FRHSZ2020consensus-easier} applied this approach along the axis of the valuation functions, while considering a general number of agents, similarly to \citep{FRG18-Consensus,FRG18-Necklace}. In this paper, we take a different approach, and we restrict the number of agents to be \emph{constant}. This is in line with most of the theoretical work on fair division, which is also concerned with solutions for a small number of agents\footnote{For example, the existence of bounded protocols for cake-cutting for more than $3$ agents was not known until very recently \citep{aziz2016discrete,aziz2016discreteb}, but such protocols for $2$ and $3$ agents were known since the 1960s (see \citep{brams1996fair,robertson1998cake}).} and it is also quite relevant from a practical standpoint, as fair division among a few participants is quite common. We believe that such investigations are necessary in order to truly understand the computational complexity of the problem. To this end, we state our first main question:

\begin{quote}
    \emph{What is the computational complexity of $\e$-\ch\ for a constant number of agents?}
\end{quote}

\noindent Since the number of agents is now fixed, any type of computational hardness must originate from the structure of the valuation functions. We remark that the existence results for $\e$-\ch\ are fairly general, and in particular do not require assumptions like additivity or monotonicity of the valuation functions. For this reason, the sensible approach is to start from valuations that are as general as possible (for which hardness is easier to establish), and gradually constrain the domain to more specific classes, until eventually polynomial-time solvability becomes possible. Indeed, in a paper that is conceptually similar to ours, \citet{deng2012algorithmic} studied the computational complexity of the \emph{contiguous envy-free cake-cutting problem}\footnote{This version of the classic envy-free cake-cutting problem requires that every agent receives a single, connected piece.} and proved that the problem is PPAD-complete, even for three agents, when agents have \emph{ordinal preferences} over the possible pieces. These types of preferences induce no structure on the valuation functions and are therefore as general as possible. In contrast, the authors showed that for three agents and \emph{monotone valuations}, the problem becomes polynomial-time solvable, leaving the case of four or more agents as an open problem.  We adopt a similar approach in this paper for $\e$-\ch, and we manage to completely settle the complexity of the problem when the agents have monotone valuations, among other results, which are highlighted in \cref{sec:ourresults}.  

Another relevant question that has been surprisingly overlooked in the related literature is the \emph{query complexity} of the problem. In this regime, the algorithm interacts with the agents via \emph{queries}, asking them to provide their values for different parts of the interval $[0,1]$, and the complexity is measured by the number of queries required to find an $\e$-approximate solution. This brings us to our second main question:

\begin{quote}
    \emph{What is the query complexity of $\e$-\ch\ for a constant number of agents?}
\end{quote}

\noindent We develop appropriate machinery that allows us to answer both of our main questions at the same time. In a nutshell, for the positive results, we design algorithms that run in polynomial time and can be recast as query-based algorithms that only use a polynomial number of queries. For the negative results, we construct reductions from ``hard'' computational problems which allow us to simultaneously obtain computational hardness results and query complexity lower bounds.

\subsection{Our Results}\label{sec:ourresults}

In this section, we list our main results regarding the computational complexity and the query complexity of the $\e$-\ch\ problem.\medskip

\noindent \textbf{Computational Complexity:} We start from the \emph{computational complexity} of the problem for a constant number of agents. We prove the following main results, parameterised by (a) the number of agents and (b) the structure of the valuation functions.

\begin{mdframed}[backgroundcolor=white!90!gray,
      leftmargin=\dimexpr\leftmargin-20pt\relax,
      innerleftmargin=0pt,
      innertopmargin=4pt,
      skipabove=5pt,skipbelow=5pt]
\begin{itemize}
    \item[-] For a \emph{single agent} and \emph{general valuations}, the problem is polynomial-time solvable. The same result applies to the case of any number of agents with \emph{identical} general valuations. 
    \item[-] For \emph{two or more agents} and \emph{general valuations}, the problem is \ppa-complete. 
    \item[-] For \emph{two agents} and \emph{monotone valuations}, the problem is polynomial-time solvable. This result holds even if one of the two agents has a general valuation.
    \item[-] For \emph{three or more agents} and \emph{monotone valuations}, the problem is \ppa-complete.
\end{itemize} 
\end{mdframed}\vskip 5pt

\noindent Finally, the $\e$-\ch\ problem with $2$ agents coincides with the well-known $\e$-\emph{Perfect Division} problem for cake-cutting (e.g., see \citep{branzei2017query,branzei2019communication}), and thus naturally our results imply that $\e$-Perfect Division with $2$ agents with monotone valuations can be done in polynomial time, whereas it becomes \ppa-complete for $2$ agents with general valuations.\medskip

\noindent Before we proceed, we offer a brief discussion on the different cases that are covered by our results. The distinction on the number of agents is straightforward. For the valuation functions, we consider mainly \emph{general valuations} and \emph{monotone valuations}. Note that neither of these functions is additive, meaning that the value that an agent has for the union of two disjoint intervals $[a,b]$ and $[c,d]$ is not necessarily the sum of her values for the two intervals. For monotone valuations, the requirement is that for any two subsets $I$ and $I'$ of $[0,1]$ such that $I \subseteq I'$, the agent values $I'$ at least as much as $I$, whereas for general valuations there is no such requirement. 

We remark that for agents with piecewise constant valuations (i.e., the valuations used in \citep{FRG18-Consensus} to obtain the \ppa-completeness of the problem for many agents), the problem can be solved rather straightforwardly in polynomial time for a constant number of agents, using linear programming (see \cref{app:piecewiseconstant}). In terms of the classification of the complexity of the problem in order of increasing generality of the valuation functions, this observation provides the ``lower bound'' whereas our \ppa-hardness results for monotone valuations provide the ``upper bound'', see \cref{fig:taxonomy}. While the precise point of the phase transition has not yet been identified, our results make considerable progress towards this goal.  \medskip

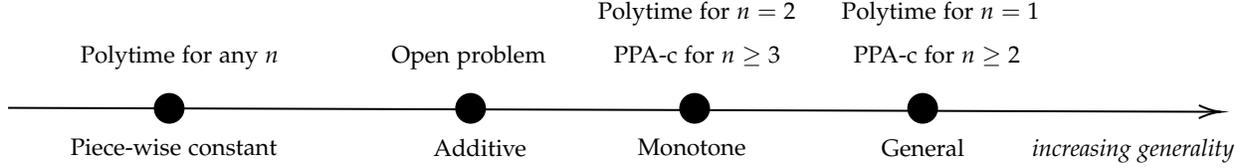
\begin{figure}
    \centering
    \scalebox{0.8}{
    \tikzset{every picture/.style={line width=0.75pt}}     

\begin{tikzpicture}[x=0.75pt,y=0.75pt,yscale=-1,xscale=1]

\draw    (24,148) -- (633.5,149.99) ;
\draw [shift={(635.5,150)}, rotate = 180.19] [color={rgb, 255:red, 0; green, 0; blue, 0 }  ][line width=0.75]    (10.93,-3.29) .. controls (6.95,-1.4) and (3.31,-0.3) .. (0,0) .. controls (3.31,0.3) and (6.95,1.4) .. (10.93,3.29)   ;

\draw  [color={rgb, 255:red, 0; green, 0; blue, 0 }  ,draw opacity=1 ][fill={rgb, 255:red, 0; green, 0; blue, 0 }  ,fill opacity=1 ] (98,148.25) .. controls (98,144.25) and (101.25,141) .. (105.25,141) .. controls (109.25,141) and (112.5,144.25) .. (112.5,148.25) .. controls (112.5,152.25) and (109.25,155.5) .. (105.25,155.5) .. controls (101.25,155.5) and (98,152.25) .. (98,148.25) -- cycle ;

\draw  [color={rgb, 255:red, 0; green, 0; blue, 0 }  ,draw opacity=1 ][fill={rgb, 255:red, 0; green, 0; blue, 0 }  ,fill opacity=1 ] (363,148.25) .. controls (363,144.25) and (366.25,141) .. (370.25,141) .. controls (374.25,141) and (377.5,144.25) .. (377.5,148.25) .. controls (377.5,152.25) and (374.25,155.5) .. (370.25,155.5) .. controls (366.25,155.5) and (363,152.25) .. (363,148.25) -- cycle ;

\draw  [color={rgb, 255:red, 0; green, 0; blue, 0 }  ,draw opacity=1 ][fill={rgb, 255:red, 0; green, 0; blue, 0 }  ,fill opacity=1 ] (478,148.25) .. controls (478,144.25) and (481.25,141) .. (485.25,141) .. controls (489.25,141) and (492.5,144.25) .. (492.5,148.25) .. controls (492.5,152.25) and (489.25,155.5) .. (485.25,155.5) .. controls (481.25,155.5) and (478,152.25) .. (478,148.25) -- cycle ;

\draw  [color={rgb, 255:red, 0; green, 0; blue, 0 }  ,draw opacity=1 ][fill={rgb, 255:red, 0; green, 0; blue, 0 }  ,fill opacity=1 ] (250,148.25) .. controls (250,144.25) and (253.25,141) .. (257.25,141) .. controls (261.25,141) and (264.5,144.25) .. (264.5,148.25) .. controls (264.5,152.25) and (261.25,155.5) .. (257.25,155.5) .. controls (253.25,155.5) and (250,152.25) .. (250,148.25) -- cycle ;

\draw (539,162) node [anchor=north west][inner sep=0.75pt]   [align=left] {{\footnotesize \textit{increasing generality}}};

\draw (54,162) node [anchor=north west][inner sep=0.75pt]   [align=left] {{\footnotesize Piecewise constant}};

\draw (340,162) node [anchor=north west][inner sep=0.75pt]   [align=left] {{\footnotesize Monotone}};

\draw (463,162) node [anchor=north west][inner sep=0.75pt]   [align=left] {{\footnotesize General}};

\draw (237,162) node [anchor=north west][inner sep=0.75pt]   [align=left] {{\footnotesize Additive}};

\draw (59,115) node [anchor=north west][inner sep=0.75pt]   [align=left] {{\footnotesize {Polytime for any $n$}}};

\draw (449,115) node [anchor=north west][inner sep=0.75pt]   [align=left] {{\footnotesize {PPA-c for $n\geq 2$}}};

\draw (328,115) node [anchor=north west][inner sep=0.75pt]   [align=left] {{\footnotesize {PPA-c for $n\geq 3$}}};

\draw (320,93) node [anchor=north west][inner sep=0.75pt]   [align=left] {{\footnotesize {Polytime for $n=2$}}};

\draw (216,115) node [anchor=north west][inner sep=0.75pt]   [align=left] {{\footnotesize Open problem}};

\draw (443,93) node [anchor=north west][inner sep=0.75pt]   [align=left] {{\footnotesize {Polytime for $n=1$}}};

\end{tikzpicture}}
    \caption{A classification of $\e$-\ch\ for a constant number $n$ of agents, in terms of increasing generality of the valuation functions.}
    \label{fig:taxonomy}
\end{figure}

\noindent \textbf{Query Complexity:} Besides the computational complexity of the problem, we are also interested in its query complexity. In this setting, one can envision an algorithm which interacts with the agents via a set of queries, and aims to compute a solution to $\e$-\ch\ using the minimum number of queries possible. In particular, a query is a question from the algorithm to an agent about a subset of $[0,1]$, who then responds with her value for that set. We provide the following results, where $L$ denotes the Lipschitz parameter of the valuation functions:

\begin{mdframed}[backgroundcolor=white!90!gray,
      leftmargin=\dimexpr\leftmargin-20pt\relax,
      innerleftmargin=0pt,
      innertopmargin=4pt,
      skipabove=5pt,skipbelow=5pt]
\begin{itemize}
    \item[-] For a \emph{single} agent and \emph{general valuations}, the query complexity of the problem is $\Theta\left(\log\frac{L}{\e}\right)$. The same result applies for any number of agents with \emph{identical} general valuations.
    \item[-] For \emph{$n \geq 2$ agents} and \emph{general valuations}, the query complexity of the problem is $\Theta\left(\left(\frac{L}{\e}\right)^{n-1}\right)$.
    \item[-] For \emph{two agents} and \emph{monotone valuations}, the query complexity of the problem is $O\left(\log^2\frac{L}{\e}\right)$. This result holds even if one of the two agents has a general valuation.
    \item[-] For \emph{$n \geq 3$ agents} and \emph{monotone valuations}, the query complexity of the problem is between $\Omega\left(\left(\frac{L}{\e}\right)^{n-2}\right)$ and $O\left(\left(\frac{L}{\e}\right)^{n-1}\right)$.
\end{itemize}
\end{mdframed}\vskip 5pt

\noindent To put these results into context, we remark that when studying the query complexity of the problem, the input consists of the error parameter $\e$ and the Lipschitz parameter $L$, given by their \emph{binary representation}. In that sense, for some constant $k$, a $\Theta(\log^k (L/\e))$ number of queries is polynomial in the size of the input. On the contrary, a $\Theta(L/\e)$ number of queries is exponential in the size of the input.
Not surprisingly, our \ppa-hardness results give rise to exponential query complexity lower bounds, whereas our algorithms can be transformed into query-based algorithms of polynomial query complexity. We remark however that beyond this connection, our query complexity analysis is in fact \emph{quantitative}, as we provide tight or almost tight bounds on the query complexity as a function of the number of agents $n$, for both general and monotone valuation functions. 

Finally, for the case of monotone valuations, we consider a more expressive query model, which is an appropriate extension of the well-known \emph{Robertson-Webb} query model \citep{robertson1998cake,woeginger2007complexity}, the predominant query model in the literature of fair cake-cutting \citep{brams1996fair}; we refer to this extension as the \emph{Generalised Robertson-Webb (GRW)} query model. We show that our bounds extend to this model as well, up to logarithmic factors.

\subsection{Related Work}

As we mentioned in the introduction, the origins of the Consensus-Halving problem can be traced back to the 1940s and the work of \citet{neyman1946theoreme}, who studied a generalisation of the problem with $k$ labels instead of two (\lplus, \lminus), and proved the existence of a solution when the valuation functions are probability measures and there is no constraint on the number of cuts used to obtain the solution. The existence theorem for two labels is known as the \emph{Hobby-Rice Theorem} \citep{hobby1965moment} and has been studied rather extensively in the context of the famous \emph{Necklace Splitting problem} \citep{Goldberg1985,Alon1986,Alon87-Necklace}. In fact, most of the proofs of existence for Necklace Splitting (with two thieves) were established via the Consensus-Halving problem, which was at the time referred to as \emph{Continuous Necklace Splitting} \citep{Alon87-Necklace}. The term ``Consensus-Halving'' was coined by \citet{SS03-Consensus}, who studied the continuous problem in isolation, and provided a constructive proof of existence which holds for very general valuation functions, including all of the valuation functions that we consider in this paper. Interestingly, the proof of \citet{SS03-Consensus} reduces the problem to finding edges of complementary labels on a triangulated $n$-sphere, labelled as prescribed by \emph{Tucker's lemma}, a fundamental result in topology.

While not strictly a reduction in the sense of computational complexity, the ideas of \citep{SS03-Consensus} certainly paved the way for subsequent work on the problem in computer science. The first computational results were obtained by \citet{filos2018hardness}, who proved that the associated computational problem, $\e$-\ch, lies in \ppa (adapting the constructive proof of \citet{SS03-Consensus}) and that the problem is PPAD-hard, for $n$ agents with piecewise constant valuation functions. \citet{FRG18-Consensus} proved that the problem is in fact \ppa-complete, establishing for the first time the existence of a ``natural'' problem complete for \ppa, i.e., a problem that does not contain a polynomial-sized circuit explicitly in its definition, answering a long-standing open question of \citet{Papadimitriou94-TFNP-subclasses}. In a follow-up paper, \citep{FRG18-Necklace} used the \ppa-completeness of Consensus-Halving to prove that the Necklace Splitting problem with two thieves is also \ppa-complete. Very recently, \citet{FRHSZ2020consensus-easier} strengthened the \ppa-hardness result to the case of very simple valuation functions, namely piecewise constant valuations with at most two blocks of value. \citet{deligkas2019computing} studied the computational complexity of the \emph{exact} version of the problem, and obtained among other results its membership in a newly introduced class BU (for ``Borsuk-Ulam'' \citep{Borsuk1933}) and its computational hardness for the well-known class FIXP of \citet{EY10-Nash-FIXP}. \citet{BatziouHH21-consensus-BBU} showed that the corresponding strong approximation problem (with measures represented by algebraic circuits) is complete for BU. A version of Consensus-Halving with divisible items was studied by \citet{GoldbergHIMS20-consensus-items}, who proved that the problem is polynomial-time solvable for additive utilities, but PPAD-hard for slightly more general utilities. Very recently, \citet{DFM-pizza} showed the PPA-completeness of the related Pizza Sharing problem \citep{karasev2016measure}, via a reduction from Consensus-Halving.

Importantly, none of the aforementioned results apply to the case of a constant number of agents, which was prior to this paper completely unexplored. Additionally, none of these works consider the query complexity of the problem. A recent work \citep{AlonG21-efficient} studies $\e$-\ch\ in a hybrid computational model (see the full version of their paper) which includes query access to the valuations, but contrary to our paper, their investigations are not targeted towards a constant number of agents, and the agents have additive valuation functions. 

A relevant line of work is concerned with the query complexity of \emph{fair cake-cutting} \citep{brams1996fair,Procaccia16-survey}, a related but markedly different fair-division problem. Conceptually closest to ours is the paper by \citet{deng2012algorithmic}, who study both the computational complexity and the query complexity of \emph{contiguous envy-free cake-cutting}, for agents with either general\footnote{More accurately, \citet{deng2012algorithmic} prove their impossibility results for \emph{ordinal preferences}, where for each possible division of the cake, the agent specifies the piece that she prefers. In particular, if one were to define valuation functions consistent with these preferences, the value of an agent for a piece would have to depend also on the way the rest of the cake is divided among the agents.} or monotone valuations. For the latter case, the authors obtain a polynomial-time algorithm for three agents, and leave open the complexity of the problem for four or more agents. In our case, for $\e$-\ch, we completely settle the computational complexity of the problem for agents with monotone valuations.

In the literature of fair cake-cutting, most of the related research (e.g., see \citep{brams1996fair,aziz2016discrete,aziz2016discreteb,amanatidis2018improved,branzei2017query}) has focused on the well-known \emph{Robertson-Webb (RW)} query model, in which agents interact with the protocol via two types of queries, \emph{evaluation queries} (\eval) and \emph{cut} queries (\cut). As the name suggests, this query model is due to \citet{robertson1995approximating,robertson1998cake}, but the work of \citet{woeginger2007complexity} has been rather instrumental in formalising it in the form that it is being used today. Given the conceptual similarity of fair cake-cutting with Consensus-Halving, it certainly makes sense to study the latter problem under this query model as well, and in fact, the queries used by \citet{AlonG21-efficient} are essentially RW queries. As we show in \cref{sec:RW}, our bounds are qualitatively robust when migrating to this more expressive model, i.e., they are preserved up to logarithmic factors.

Related to our investigation is also the work of \citet{branzei2017query}, who among other settings study the problem of \emph{$\e$-Perfect Division}, which stipulates a partition of the cake into $n$ pieces, such that each of the $n$ agents interprets \emph{all} pieces to be of approximate equal value (up to $\e$). For the case of $n=2$, this problem coincides with $\e$-\ch, and thus one can interpret our results for $n=2$ as extensions of the results in \citep{branzei2017query} (which are only for additive valuations) to the case of monotone valuations (for which the problem is solvable with polynomially-many queries) and to the case of general valuations (for which the problem admits exponential query complexity lower bounds). Besides the aforementioned results, there is a plethora of works in computer science and artificial intelligence related to computational aspects of fair cake cutting and fair division in general, e.g., see \citep{elkind2021mind,segal2018redividing,bei2017cake,bei2012optimal,bei2021price,bei2021dividing,goldberg2020contiguous,balkanski2014simultaneous,alijani2017envy,kurokawa2013cut,branzei2016algorithmic,cohler2011optimal,hosseini2020fair}.

\section{Preliminaries}\label{sec:preliminaries}

In the $\e$-\ch\ problem, there is a set of $n$ agents with \emph{valuation functions} $v_i$ (or simply \emph{valuations}) over the interval $[0,1]$, and the goal is to find a partition of the interval into subintervals labelled either \lplus or \lminus, using at most $n$ cuts. This partition should satisfy that for every agent $i$, the total value for the union of subintervals $\calI^{+}$ labelled \lplus and the total value for the union of subintervals $\mathcal{I}^{-}$ labelled \lminus is the same up to $\e$, i.e., $|v_i(\calI^{+}) - v_i(\calI^{-})| \leq \e$. In this paper we will assume $n$ to be a constant and therefore the inputs to the problem will only be $\e$ and the valuation functions $v_i$.

We will be interested in fairly general valuation functions; intuitively, these will be functions mapping \emph{measurable subsets} $A \subseteq [0,1]$ to non-negative real numbers. Formally, let $\Lambda([0,1])$ denote the set of Lebesgue-measurable subsets of the interval $[0,1]$ and $\lambda: \Lambda([0,1]) \to [0,1]$ the Lebesgue measure. We consider valuation functions $v_i: \Lambda([0,1]) \to \mathbb{R}_{\geq 0}$, with the interpretation that agent $i$ has value $v_i(A)$ for the subset $A \in \Lambda([0,1])$ of the resource. Similarly to \citep{deng2012algorithmic,branzei2017query,branzei2019communication,BarmanR20cake,AlonG21-efficient}, we also require that the valuation functions be Lipschitz-continuous. Following \citep{deng2012algorithmic}, a valuation function $v_i$ is said to be Lipschitz-continuous with Lipschitz parameter $L \geq 0$, if for all $A,B \in \Lambda([0,1])$, it holds that $|v_i(A) - v_i(B)| \leq L \cdot \lambda(A \triangle B)$. Here $\triangle$ denotes the symmetric difference, i.e., $A \triangle B = (A \setminus B) \cup (B \setminus A)$.\medskip

\noindent \textbf{Valuation Classes:} We will be particularly interested in the following three valuation classes, in terms of decreasing generality: 
\begin{itemize}
    \item[-] \emph{General valuations}, in which there is no further restriction to the functions $v_i$.
    \item[-] \emph{Monotone valuations}, in which $v_i(A) \leq v_i(A')$ for any two Lebesgue-measurable subsets $A$ and $A'$ such that $A \subseteq A'$. Intuitively, for this type of function, when comparing two sets such that one is a subset of the other, an agent cannot have a smaller value for the set that contains the other.
    \item[-] \emph{Additive valuations}, in which $v_i$ is a function from individual intervals in $[0,1]$ to  $\mathbb{R}_{\geq 0}$ and for a set of intervals $\calI$, it holds that $v_i(\calI)=\sum_{I \in \calI} v_i(I)$. Note that if $v_i$ is an additive valuation function, then it is in fact a measure. We will not prove any results for this type of valuation function in this paper, but we define them for reference and comparison (e.g., see \cref{fig:taxonomy}).
\end{itemize}
\medskip

\noindent \textbf{Normalisation:} We will also be interested in valuation functions that satisfy some standard normalisation properties. A valuation function $v_i$ is \emph{normalised}, if the following properties hold:
\begin{enumerate}
    \item $v_i(A) \in [0,1]$ for all $A \in \Lambda([0,1])$,
    \item $v_i(\emptyset) = 0$ and $v_i([0,1]) = 1$.
\end{enumerate}
In other words, we require the agents' values to lie in $[0,1]$ and that their value for the whole interval is normalised to $1$. These are the standard assumptions in the literature of the problem for additive valuations \citep{Alon87-Necklace}, as well as in the related problem of fair cake-cutting \citep{Procaccia16-survey}. We will only consider normalised valuation functions for our lower bounds and hardness results, whereas for the upper bounds and polynomial-time algorithms we will not impose any normalisation; this only makes both sets of results even stronger. 

With regard to the valuation classes defined above, we will often be referring to their normalised versions as well, e.g., \emph{normalised general valuations} or \emph{normalised monotone valuations}.\medskip

\noindent \textbf{Input models:} Given the fairly general nature of the valuation functions, we need to specify the manner in which they will be accessed by an algorithm for $\e$-\ch. Since we are interested in both the computational complexity and the query complexity of the problem, we will assume the following standard two ways of accessing these functions.
\begin{itemize}
    \item[-] In the \textbf{\emph{black-box model}}, the valuation functions $v_i$ can be \emph{arbitrary functions}, and are accessed via \emph{queries} (sometimes also referred to as \emph{oracle calls}). A query to the function $v_i$ inputs a Lebesgue-measurable subset $A$ (intuitively a set of subintervals) of $[0,1]$ and outputs $v_i(A) \in \mathbb{R}_{\geq 0}$. This input model is appropriate for studying the \emph{query complexity} of the problem, where the complexity is measured as the number of queries to the valuation function $v_i$. \smallskip
    
    \noindent We will also consider the following weaker version of the black-box model, which we will use in our query complexity upper bounds, thus making them stronger: In the \textbf{\emph{weak black-box model}} the input to a valuation function $v_i$ is some set $\calI$ of intervals, obtained by using at most $n$ cuts, where $n$ is the number of agents.
    \item[-] In the \textbf{\emph{white-box model}}, the valuation functions $v_i$ are \emph{polynomial-time algorithms}, mapping sets of intervals to non-negative rational numbers. These polynomial-time algorithms are given explicitly as part of the input, including the Lipschitz parameter $L$.\footnote{To avoid introducing too many technical details here, we refer to \cref{app:preliminaries} for the fully formal definition.} This input model is appropriate for studying the \emph{computational complexity} of the problem, where the complexity is measured as usual by the running time of the algorithm.
\end{itemize}

\noindent We now provide the formal definitions of the problem in the black-box and the white-box model. Note that the Lipschitz-parameter $L$ is part of the input of the problem and thus will appear in the bounds we obtain. Some of the previous works take $L$ to be bounded by a constant, and as a result it does not appear in their bounds.

\begin{mdframed}[backgroundcolor=white!90!gray,
      leftmargin=\dimexpr\leftmargin-20pt\relax,
      innerleftmargin=4pt,
      innertopmargin=0pt,
      skipabove=5pt,skipbelow=5pt]
 \begin{definition}[\textup{$\e$-\ch} \textit{(black-box model)}]
 For any constant $n \geq 1$, the problem \emph{$\e$-\ch\ with $n$ agents} is defined as follows:
 \begin{itemize}
     \item[-] \textbf{Input:} $\e >0$, the Lipschitz parameter $L$, \emph{query access} to the functions $v_i$.
     \item[-] \textbf{Output:} A partition of $[0,1]$ into two sets of intervals $\calI^{+}$ and $\calI^{-}$ such that for each agent $i$, it holds that $|v_i(\calI^{+}) - v_i(\calI^{-})| \leq \e$, using at most $n$ cuts.
 \end{itemize}
 \end{definition}
 \end{mdframed}

 \begin{mdframed}[backgroundcolor=white!90!gray,
      leftmargin=\dimexpr\leftmargin-20pt\relax,
      innerleftmargin=4pt,
      innertopmargin=0pt,
      skipabove=5pt,skipbelow=5pt]
 \begin{definition}[\textup{$\e$-\ch} \textit{(white-box model)}]
 For any constant $n \geq 1$, the problem \emph{$\e$-\ch\ with $n$ agents} is defined as follows:
 \begin{itemize}
     \item[-] \textbf{Input:} $\e >0$, the Lipschitz parameter $L$, polynomial-time algorithms $v_i$.
     \item[-] \textbf{Output:} A partition of $[0,1]$ into two sets of intervals $\calI^{+}$ and $\calI^{-}$ such that for each agent $i$, it holds that $|v_i(\calI^{+}) - v_i(\calI^{-})| \leq \e$, using at most $n$ cuts.
 \end{itemize}
 \end{definition}
 \end{mdframed}\vskip 5pt
 
 \noindent \textbf{Terminology:} When the valuation functions are normalised, we will refer to the problem as \emph{$\e$-\ch\ with $n$ normalised agents}. When the valuation functions are monotone, we will refer to the problem as \emph{$\e$-\ch\ with $n$ monotone agents}. If both conditions are true, we will use the term \emph{$\e$-\ch\ with $n$ normalised monotone agents}.

\subsection{\bu and \tucker}\label{sec:ppa}

For our \ppa-hardness results and query complexity lower bounds, we will reduce from a well-known problem, the computational version of the Borsuk-Ulam Theorem \citep{Borsuk1933}, which states that for any continuous function $F$ from $S^n$ to $\mathbb{R}^n$ there is a pair of antipodal points (i.e., $x,-x$) which are mapped to the same point. There are various equivalent versions of the problem (e.g., see \citep{Mat03BorsukUlam}); we will provide a definition that is most appropriate for our purposes. In fact, we will include several ``optional'' properties of the function $F$ in our definition, which will map to properties of the valuation functions $v_i$ when we construct our reductions in subsequent sections. Specifically, we will impose conditions for \emph{normalisation} and \emph{monotonicity}, which will correspond to normalised valuation functions for our lower bounds/hardness results of \cref{sec:general}, and to normalised monotone valuation functions for our lower bounds/hardness results of \cref{sec:monotone}. Let $B^n = [-1,1]^n$ and let $\partial(B^n)$ denote its boundary. As before, we will require that the functions we consider be Lipschitz-continuous. We say that $F: B^{n+1} \to B^n$ is Lipschitz-continuous with parameter $L$, if $\|F(x)-F(y)\|_\infty \leq L \cdot \|x-y\|_\infty$ for all $x,y \in B^{n+1}$, where $\|x\|_\infty = \max_i |x_i|$.

\begin{mdframed}[backgroundcolor=white!90!gray,
      leftmargin=\dimexpr\leftmargin-20pt\relax,
      innerleftmargin=4pt,
      innertopmargin=0pt,
      skipabove=5pt,skipbelow=5pt]
 \begin{definition}[\textup{$nD$-\bu}]\label{def:BU}
 For any constant $n \geq 1$, the problem $nD$-\bu is defined as follows:
 \begin{itemize}
     \item[-] \textbf{Input:} $\e >0$, the Lipschitz parameter $L$, a function $F: B^{n+1} \to B^n$.
     \item[-] \textbf{Output:} A point $x \in \partial (B^{n+1})$ such that $\|F(x) - F(-x)\|_\infty \leq \e$.
     \item[-] \textbf{Optional Properties:}
     \begin{itemize}
        \item[] \emph{\underline{Normalisation:}}
         \item[\tiny $\blacksquare$] $F(1,1,\ldots,1)=(1,1,\ldots,1)$.
         \item[\tiny $\blacksquare$] $F(-x) = -F(x)$, for all $x \in B^{n+1}$. \medskip
         \item[] \emph{\underline{Monotonicity:}}
         \item[\tiny $\blacksquare$] If $x \leq y$, then $F(x) \leq F(y)$, for all $x, y \in B^{n+1}$, where ``$\leq$'' denotes coordinate-wise comparison.
     \end{itemize}
 \end{itemize}
 \end{definition}
 \end{mdframed}\vskip 5pt
 
\noindent In the \emph{normalised $nD$-\bu problem}, where $F$ is normalised, we instead ask for a point $x \in \partial (B^{n+1})$ such that $\|F(x)\|_\infty \leq \varepsilon$. By using the fact that $F$ is an odd function ($F(-x) = -F(x)$, for all $x \in B^{n+1}$), it is easy to see that this is equivalent to $\|F(x) - F(-x)\|_\infty \leq \e/2$. We will also use the term \emph{normalised monotone $nD$-\bu} to refer to the problem when both the normalisation and monotonicity properties are satisfied for $F$.

In the black-box version of $nD$-\bu, we can query the value of the function $F$ at any point $x \in B^{n+1}$. In the white-box version of this problem, we are given a polynomial-time algorithm that computes $F$. Since the number of inputs of $F$ is fixed, we can assume that we are given an arithmetic circuit with $n+1$ inputs and $n$ outputs that computes $F$. Following the related literature \citep{DaskalakisP11-CLS}, we will consider circuits that use the arithmetic gates $+,-,\times, \max, \min, <$ and rational constants.\footnote{Formally speaking these circuits also need to be \emph{well-behaved} in a certain sense (see \citep{FearnleyGHS-gradient}). It is easy to check that the circuits we construct in this paper all have this property.}\\

\noindent Another related problem that will be of interest to us is the computational version of Tucker's Lemma \citep{tucker1945some}. Tucker's lemma is a discrete analogue of the Borsuk-Ulam theorem, and its computational counterpart, $nD$-\tucker, is defined below.

\begin{mdframed}[backgroundcolor=white!90!gray,
      leftmargin=\dimexpr\leftmargin-20pt\relax,
      innerleftmargin=4pt,
      innertopmargin=0pt,
      skipabove=5pt,skipbelow=5pt]
 \begin{definition}[\textup{$nD$-\tucker}]\label{def:tucker}
 For any constant $n \geq 1$, the problem $nD$-\tucker is defined as follows:
 \begin{itemize}
     \item[-] \textbf{Input:} grid size $N \geq 2$, a labelling function $\ell: [N]^n \to \{\pm 1, \pm 2, \dots, \pm n\}$ that is antipodally anti-symmetric (i.e., for any point $p$ on the boundary of $[N]^n$, we have $\ell(\overline{p}) = - \ell(p)$, where $\overline{p}_i = N+1-p_i$ for all $i$).
     \item[-] \textbf{Output:} Two points $p,q \in [N]^n$ with $\ell(p) = -\ell(q)$ and $\|p-q\|_\infty \leq 1$.
 \end{itemize}
 \end{definition}
 \end{mdframed}\vskip 5pt
 
 \noindent In the black-box version of this problem, we can query the labelling function for any point $p \in [N]^n$ and retrieve its label. In the white-box version, $\ell$ is given in the form of a Boolean circuit with the usual gates $\land, \lor, \lnot$.\medskip
 
 \noindent In the white-box model, $nD$-\tucker was recently proven to be \ppa-hard for any $n \geq 2$, by \citet{ABB15-2DTucker}; the membership of the problem in \ppa was known by \citep{Papadimitriou94-TFNP-subclasses}. 
 
 \begin{theorem}[\citep{ABB15-2DTucker,Papadimitriou94-TFNP-subclasses}]\label{lem:tuckerppa}
For any constant $n \geq 2$, \textup{$nD$-\tucker}\ is \ppa-complete.
\end{theorem}

\noindent The computational class \ppa\ was defined by \citet{Papadimitriou94-TFNP-subclasses}, among several subclasses of the class TFNP \citep{Megiddo1991}, the class of problems with a guaranteed solution which is verifiable in polynomial time.
PPA is defined with respect to a graph of exponential size, which is given \emph{implicitly} as input, via the use of a circuit that outputs the neighbours of a given vertex, and the goal is to find a vertex of odd degree, given another such vertex as input. 

Given the close connection between Tucker's Lemma and the Borsuk-Ulam Theorem, the \ppa-hardness of \emph{some appropriate computational version} of Borsuk-Ulam follows as well \citep{ABB15-2DTucker}. However, this does not apply to the version of $nD$-\bu defined above, especially when one considers the additional properties of the function $F$ required for normalisation and monotonicity, as discussed earlier. We will reduce from $nD$-\tucker to our version of $nD$-\bu, to obtain its \ppa-hardness (even for the normalised monotone version), which will then imply the \ppa-hardness of $\e$-\ch, via our main reduction in \cref{sec:black-box}. \medskip

\noindent In the black-box model, \citet{deng2011discrete}, building on the results of \citet{chen2008matching}, proved both query complexity lower bounds and upper bounds for $nD$-\tucker.

\begin{theorem}[\citep{deng2011discrete}]\label{lem:tucker-query}
For any constant $n \geq 2$, the query complexity of \textup{$nD$-\tucker}\ is $\Theta \left(N^{n-1}\right)$.
\end{theorem}

\noindent We remark that \citet{deng2011discrete} use a version of $nD$-\tucker that is slightly different from the one that we defined above, but their results apply to this version as well.

\begin{remark}
Some connections between the aforementioned problems, namely $nD$-\bu, $nD$-\tucker, and $\e$-\ch\ are known from the previous literature. First, $nD$-\bu\ and $nD$-\tucker\ are known to be computationally equivalent due to \citet{Papadimitriou94-TFNP-subclasses} and \citet{ABB15-2DTucker}, although, technically speaking, none of the aforementioned papers proved this result formally, or even defined $nD$-\bu\ formally. Yet, even the implicit reduction between those problems is insufficient for our purposes, because, in order to achieve our results for $\ch$, we need a version of $nD$-\bu\ that exhibits several properties, as described in \cref{def:BU} (namely, normalisation and primarily monotonicity). Indeed, proving the PPA-completeness of monotone $nD$-\bu\ is the main technical result of our work. 

In terms of the completeness of $\ch$, the works of \citet{FRG18-Consensus,FRG18-Necklace} indeed establish reductions from $nD$-\tucker, but crucially, these reductions are white-box, i.e., they have access to the Boolean circuit that encodes the labelling function. On the contrary, our reductions are black-box (see \cref{sec:black-box-reductions} below), which allows us to obtain both computational complexity results and query complexity bounds at the same time. Also, quite importantly, all the previous reductions do not work when there is a constant number of agents.
\end{remark}

\subsection{Efficient Black-Box Reductions}\label{sec:black-box-reductions}

The reductions that we will construct (from $nD$-\tucker to $nD$-\bu to $\e$-\ch) will be \emph{black-box reductions}, and therefore they will also allow us to obtain query complexity lower bounds for $\e$-\ch\ in the black-box model, given the corresponding lower bounds of \cref{lem:tucker-query}. For the upper bounds, we will reduce directly from $\e$-\ch\ to $nD$-\tucker, again via a black-box reduction. 

Roughly speaking,\footnote{To keep the exposition clean, we do not provide formal definitions of these concepts here, as they are rather standard; we refer the reader to the related works of \citep{Beame1998,komargodski2019white} for more details.} a black-box reduction from Problem A to Problem B is a procedure by which we can answer oracle calls (queries) for an instance of Problem B by using an oracle for some instance of Problem A, such that a solution to the instance of Problem B yields a solution to the instance of Problem A. For example, a black-box reduction from $nD$-\bu to $\e$-\ch\ is a procedure that simulates an instance for the latter problem by accessing the function $F$ of the former problem a number of times, and such that a solution of $\e$-\ch\ can easily be translated to a solution of $nD$-\bu. The name ``black-box'' comes from the fact that this type of reduction does not need to know the structure of the functions $v_i$ of $\e$-\ch\ or $F$ of $nD$-\bu.

In order to prove lower bounds on the query complexity of some Problem B, it suffices to construct a black-box reduction from some Problem A, for which query complexity lower bounds are known; the obtained bounds will depend on the number $k$ of oracle calls to the input of Problem A that are needed to answer an oracle call to the input of Problem B. A black-box reduction is \emph{efficient} if $k$ is a constant, and therefore the query complexity lower bounds of Problems A and B are of the same asymptotic order. To obtain upper bounds on the query complexity, we can construct a reduction in the opposite direction (from Problem B to Problem A), assuming that query complexity upper bounds for Problem A are known.

Ideally, we would like to use the same reduction to also obtain computational complexity results in the white-box model. For this to be possible, the procedure described above should actually be a polynomial-time algorithm. Slightly abusing terminology, we will use the term ``efficient'' to describe such a reduction in the white-box model as well.

\begin{definition}
We say that a black-box reduction from Problem A to Problem B is \emph{efficient} if:
\begin{itemize}
    \item[-] in the black-box model, it uses a constant number of queries (oracle calls) to the function (oracle) of Problem A, for each query (oracle call) to the function of Problem B;
    \item[-] in the white-box model, the condition above holds, and the reduction is also a polynomial-time algorithm.
\end{itemize}
\end{definition}

\noindent Concretely for our case, all of our reductions will be efficient black-box reductions, thus allowing us to obtain both \ppa-completeness results and query complexity bounds matching those of the problems that we reduce from/to. We remark that the reductions constructed for proving the \ppa-hardness of the problem in previous works (for a non-constant number of agents) \citep{FRG18-Consensus,FRG18-Necklace,FRHSZ2020consensus-easier} are not black-box reductions, and therefore have no implications on the query complexity of the problem.

\section{Black-Box Reductions to and from Consensus-Halving}\label{sec:black-box}

In this section we develop our main machinery for proving both $\ppa$-completeness results and query complexity upper and lower bounds for $\e$-\ch. 
We summarise our general approach for obtaining positive and negative results below.\medskip

\noindent For our \emph{impossibility results} (i.e., computational hardness results in the white-box model and query complexity lower bounds in the black-box model), we will construct an efficient black-box reduction from $nD$-\bu to $\e$-\ch\ with $n$ agents (\cref{prop:BU2CH}). This reduction will preserve the optional properties of \cref{def:BU}, meaning that if the instance of $nD$-\bu is normalised (respectively monotone), the valuation functions of the corresponding instance of $\e$-\ch\ will be normalised (respectively monotone) as well. This will allow us in subsequent sections to reduce the problem of proving impossibility results for $\e$-\ch\ to proving impossibility results for the versions of $nD$-\bu with those properties. We will obtain these latter results via reductions from $nD$-\tucker, which for $n \geq 2$ is known to be \ppa-hard (\cref{lem:tuckerppa}) and admit exponential query complexity lower bounds (\cref{lem:tucker-query}). \vskip 3pt

\noindent For our \emph{positive results} (i.e., membership in \ppa in the white-box model and query complexity upper bounds in the black-box model), we will construct an efficient black-box reduction from $\e$-\ch\ to $nD$-\tucker (\cref{prop:CH2TU}). We remark here that a similar reduction already exists in the related literature \citep{filos2020topological}, but only applied to the case of additive valuation functions. The extension to the case of general valuations follows along the same lines, and we provide it here for completeness. We also note that some of our positive results, namely the results for one general agent and two monotone agents, will not be obtained via reductions, but rather directly via the design of polynomial-time algorithms in the white-box model or algorithms of polynomial query complexity in the black-box model. \medskip

\noindent Related to the discussion above, we have the following two propositions. Their proofs are presented in the sections below.

\begin{proposition}\label{prop:BU2CH}
There is an efficient black-box reduction from (normalised, monotone) $nD$-\bu\ to (normalised, monotone) $\e$-\ch. 
\end{proposition}

\begin{proposition}\label{prop:CH2TU}
There is an efficient black-box reduction from $\e$-\ch\ to $nD$-\tucker.
\end{proposition}

\subsection{Proof of Proposition \ref{prop:BU2CH}}
\noindent\textbf{Description of the reduction.}
Let $n \geq 1$ be a fixed integer. Let $\varepsilon > 0$ and let $F: B^{n+1} \to B^n$ be a Lipschitz-continuous function with Lipschitz parameter $L$. We now construct valuation functions $v_1, \dots, v_n$ for a Consensus-Halving instance.

Let $R_1, R_2, \dots, R_{n+1}$ denote the partition of interval $[0,1]$ into $n+1$ subintervals of equal length, i.e., $R_j = [\frac{j-1}{n+1}, \frac{j}{n+1}]$ for $j \in [n+1]$. For any $A \in \Lambda([0,1])$, we define $x(A) \in B^{n+1} = [-1,1]^{n+1}$ by
$$[x(A)]_j = 2 (n+1) \cdot \lambda(A \cap R_j) - 1$$
for all $j \in [n+1]$. Recall that $\lambda$ denotes the Lebesgue measure on the interval $[0,1]$. Note that since $\lambda(A \cap R_j) \in [0,\frac{1}{n+1}]$, we indeed have $[x(A)]_j \in [-1,1]$; see \cref{fig:regions} for a visualisation.

For $i \in [n]$, the valuation function $v_i$ of the $i$th agent is defined as
$$v_i(A) = \frac{F_i(x(A)) + 1}{2}$$
for any $A \in \Lambda([0,1])$, where $F_i: B^{n+1} \to [-1,1]$ is the $i$th output of $F$. Note that $v_i(A) \in [0,1]$, since $F_i(x(A)) \in [-1,1]$.\medskip

\noindent\textbf{Lipschitz-continuity.}
For any $A,B \in \Lambda([0,1])$ it holds that
\begin{align*}
|v_i(A) - v_i(B)|  & = \frac{1}{2} \left|F_i(x(A)) - F_i(x(B))\right| \leq \frac{1}{2} \|F(x(A)) - F(x(B))\|_\infty \\
& \leq \frac{L}{2} \|x(A) - x(B)\|_\infty  \leq (n+1) \cdot L \cdot \max_{j \in [n+1]} |\lambda(A \cap R_j) - \lambda(B \cap R_j)| \\
& \leq (n+1)  \cdot L \cdot \max_{j \in [n+1]} \lambda((A \triangle B) \cap R_j) \leq (n+1)  \cdot L \cdot \lambda(A \triangle B).
\end{align*}
Thus, $v_i$ is Lipschitz-continuous with Lipschitz parameter $(n+1)  \cdot L$.\medskip

\noindent{\bf Correctness.}
Now consider an $\varepsilon/2$-Consensus-Halving of $v_1, \dots, v_n$. Namely, let $\mathcal{I}^+, \mathcal{I}^-$ be a partition of $[0,1]$ using at most $n$ cuts, such that $|v_i(\mathcal{I}^+) - v_i(\mathcal{I}^-)| \leq \varepsilon/2$ for all $i \in [n]$. Since $\mathcal{I}^+$ is obtained by using at most $n$ cuts, it follows that there exists $\ell \in [n+1]$ such that $R_\ell$ does not contain a cut. As a result, $\mathcal{I}^+ \cap R_\ell$ is either empty or equal to $R_\ell$. This implies that $\lambda(\mathcal{I}^+ \cap R_\ell) \in \{0, \frac{1}{n+1}\}$ and thus $[x(\mathcal{I}^+)]_\ell \in \{\pm 1\}$, i.e., $x(\mathcal{I}^+) \in \partial (B^{n+1})$. Furthermore, for any $j \in [n+1]$, we have
\begin{equation*}
\begin{split}
[x(\mathcal{I}^+)]_j & = 2 (n+1) \cdot \lambda(\mathcal{I}^+ \cap R_j) - 1 = 2 (n+1) \cdot \left(\frac{1}{n+1} - \lambda(\mathcal{I}^- \cap R_j)\right) - 1\\
&= -2 (n+1) \cdot \lambda(\mathcal{I}^- \cap R_j) + 1
= -[x(\mathcal{I}^-)]_j.
\end{split}
\end{equation*}
Letting $y = x(\mathcal{I}^+) \in \partial (B^{n+1})$, we have that for any $i \in [n]$
$$|F_i(y) - F_i(-y)| = |F_i(x(\mathcal{I}^+)) - F_i(x(\mathcal{I}^-))| = 2 |v_i(\mathcal{I}^+) - v_i(\mathcal{I}^-)| \leq \varepsilon.$$
Thus, $y$ is a solution to the original $nD$-\bu instance.\medskip

\noindent\textbf{White-box model.} This reduction yields a polynomial-time many-one reduction from $nD$-\bu to $\e$-\ch\ with $n$ agents. Thus, if we show that $nD$-\bu is PPA-hard for some $n$, then we immediately obtain that $\e$-\ch\ with $n$ agents is also PPA-hard.\medskip

\noindent\textbf{Black-box model.} It is easy to see that this is a black-box reduction. It can be formulated as follows: given access to an oracle for an instance of $nD$-\bu with parameters $(\varepsilon, L)$ we can simulate an oracle for an instance of $\e$-\ch\ (with $n$ agents) with parameters $(\varepsilon/2, (n+1) L)$ such that any solution of the latter yields a solution to the former. Furthermore, in order to answer a query to some $v_i$, we only need to perform a single query to $F$. Thus, we obtain the following query lower bound: solving an instance of $\e$-\ch\ (with $n$ agents) with parameters $(\varepsilon, L)$ requires at least as many queries as solving an instance of $nD$-\bu with parameters $(\varepsilon',L') = (2\varepsilon,\frac{L}{n+1})$. This means that if $nD$-\bu has a query lower bound of $\Omega((\frac{L'}{\varepsilon'})^{n-1})$ for some $n$, then $\e$-\ch\ (with $n$ agents) has a query lower bound of $\Omega((\frac{L}{2 \varepsilon (n+1)})^{n-1}) = \Omega((\frac{L}{\varepsilon})^{n-1})$, since $n$ is constant.\medskip

\begin{figure}
  \begin{center}
      \includegraphics[width=\textwidth]{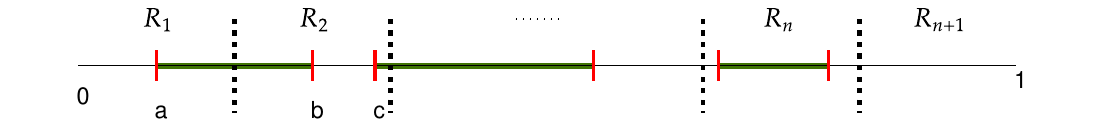}
  \end{center}
    \caption{The partition of $[0,1]$ into $n+1$ subintervals of equal length and a set $A$, coloured by the green region, as it is defined by the red cuts. The first three cuts on the left are located at positions $a \leq b \leq c$, where $a \in R_1$ and $b,c \in R_2$. Here, since $\lambda(A \cap R_1) = 1/(n+1)-a$, we would obtain $[x(A)]_1 = 2(n+1)(1/(n+1)-a) -1 = 1 - 2(n+1)a$. Similarly, $\lambda(A \cap R_2) = 1/(n+1) - (c-b)$, and thus $[x(A)]_2 = 1 - 2(n+1)(c-b)$.}
    \label{fig:regions}
\end{figure}

\noindent\textbf{Additional properties of the reduction.} Some properties of the Borsuk-Ulam function $F$ carry over to the valuation functions $v_1, \dots, v_n$. In particular, the following properties are of interest to us:
\begin{itemize}
    \item[-] \textbf{If $F$ is monotone, then $v_1, \dots, v_n$ are monotone.} Indeed, consider $A, B \in \Lambda([0,1])$ with $A \subseteq B$. Then, it holds that $\lambda(A \cap R_j) \leq \lambda(B \cap R_j)$ for all $j \in [n+1]$, and as a result $x(A) \leq x(B)$ (coordinate-wise). By monotonicity of $F$, it follows that $F_i(x(A)) \leq F_i(x(B))$, and thus $v_i(A) \leq v_i(B)$ for all $i \in [n]$.
    \item[-] \textbf{If $F$ is normalised, then $v_1, \dots, v_n$ are normalised.} As noted earlier, we already have that $v_i(A) \in [0,1]$ for all $A \in \Lambda([0,1])$. Thus, it remains to prove that $v_i(\emptyset) = 0$ and $v_i([0,1]) = 1$. It is easy to see that $x([0,1]) = (1, 1, \dots, 1)$ and thus $F(x([0,1])) = (1, 1, \dots, 1)$ since $F$ is normalised, which yields $v_i([0,1]) = 1$. On the other hand, we have $x(\emptyset) = (-1,-1,\dots,-1)$ and thus $F(x(\emptyset)) = -F(-x(\emptyset)) = - F(1,1,\dots,1) = (-1, -1, \dots, -1)$, which yields $v_i(\emptyset) = 0$. Here we also used the fact $F$ is an odd function, since it is normalised. In fact, since $F$ is odd, we also obtain that $v_i(A) + v_i(A^c) = 1$ for all $A \in \Lambda([0,1])$, where $A^c = [0,1] \setminus A$ denotes the complement of $A$. This can be shown by noting that $x(A^c) = -x(A)$ (by using the same argument as for $\mathcal{I}^+$ and $\mathcal{I}^-$ above) and then using the fact that $F(x(A^c)) = -F(x(A))$. 
\end{itemize}
This means that if we are able to show (white- or black-box) hardness of $nD$-\bu where $F$ has additional properties, then the hardness will also hold for $\e$-\ch\ with $n$ agents that have the corresponding properties.

Furthermore, note that if $F$ is a normalised $nD$-\bu function, then an $\varepsilon$-approximate Consensus-Halving for $v_1, \dots, v_n$ (i.e., $|v_i(\mathcal{I}^+) - v_i(\mathcal{I}^-)| \leq \varepsilon$), yields an $\varepsilon$-approximate solution to $F$, in the sense that $\|F(x(\mathcal{I}^+))\|_\infty \leq \varepsilon$. This is due to the fact that, by definition, $F$ is an odd function, if it is normalised.

\subsection{Proof of Proposition \ref{prop:CH2TU}}\label{sec:CHtoTU}

The reduction presented in this section is based on a proof of existence for a generalization of Consensus-Halving with general valuations given in \citep{filos2020topological}. This existence proof was also used in \citep{filos2020topological} to provide a reduction for additive valuations. It can easily be extended to work for general valuations as well. We include the full reduction here for completeness.\medskip

\noindent\textbf{Description of the reduction.}
Consider an instance of $\varepsilon$-\ch\ with $n$ agents with parameters $\varepsilon$, $L$. Let $v_1, \dots, v_n$ denote the valuations of the agents. We consider the domain $K_m^n$, where $K_m = \{-1,-(m-1)/m,\dots, -1/m,0,1/m,2/m,\dots, (m-1)/m,1\}$, for $m=\lceil 2nL/\varepsilon \rceil$. A point in $K_m^n$ corresponds to a way to partition the interval $[0,1]$ into two sets $\mathcal{I}^+, \mathcal{I}^-$ using at most $n$ cuts. A very similar encoding was also used by \citet{Meunier2014simplotopal} for the Necklace Splitting problem. A point $x \in K_m^n$ corresponds to the partition $\mathcal{I}^+(x), \mathcal{I}^-(x)$ obtained as follows.
\begin{enumerate}
    \item Provisionally put the label \lplus on the whole interval $[0,1]$
    \item For $\ell=1,2,\dots,n$:
    \begin{itemize}
        \item[-] if $x_\ell > 0$, then put label \lplus on the interval $[0,x_\ell]$;
        \item[-] if $x_\ell < 0$, then put label \lminus on the interval $[0,-x_\ell]$.
    \end{itemize}
\end{enumerate}
Note that subsequent assignments of a label to an interval, ``overwrite'' previous assignments. One way of thinking about it, is that we are applying a coat of paint on the interval $[0,1]$. Initially the whole interval is painted with colour \lplus, and as the procedure is executed, various subintervals will be painted over with colour \lminus or \lplus. It is easy to check that the final partition into $\mathcal{I}^+(x), \mathcal{I}^-(x)$ that is obtained, uses at most $n$ cuts. Furthermore, for any $x \in \partial K_m^n$, the partition $\mathcal{I}^+(-x), \mathcal{I}^-(-x)$ obtained from $-x$ corresponds to the partition $\mathcal{I}^+(x), \mathcal{I}^-(x)$ with labels \lplus and \lminus switched. In other words, $\mathcal{I}^+(-x) = \mathcal{I}^-(x)$ and $\mathcal{I}^-(-x) = \mathcal{I}^+(x)$. For a more formal definition of this encoding, see \citep{filos2020topological}.

We define a labelling $\ell: K_m^n \to \{\pm 1, \pm 2, \dots, \pm n\}$ as follows. For any $x \in K_m^n$:
\begin{enumerate}
    \item Let $i \in [n]$ be the agent that sees the largest difference between $v_i(\mathcal{I}^+(x))$ and $v_i(\mathcal{I}^-(x))$, i.e., $i = \arg\max_{i \in [n]} |v_i(\mathcal{I}^+(x)) - v_i(\mathcal{I}^-(x))|$, where we break ties by picking the smallest such $i$.
    \item Pick a sign $s \in \{+,-\}$ as follows. If $v_i(\mathcal{I}^+(x)) > v_i(\mathcal{I}^-(x))$, then let $s=+$. If $v_i(\mathcal{I}^+(x)) < v_i(\mathcal{I}^-(x))$, then let $s=-$. If $v_i(\mathcal{I}^+(x)) = v_i(\mathcal{I}^-(x))$, then pick $s$ such that $\mathcal{I}^s$ contains the left end of the interval $[0,1]$.
    \item Set $\ell(x) = +i$ if $s=+$, and $\ell(x) = -i$ otherwise.
\end{enumerate}

With this definition, it is easy to check that $\ell(-x) = -\ell(x)$ for all $x \in \partial(K_m^n)$. By re-interpreting $K_m^n$ as a grid $[N]^n$ with $N = 2m+1$, we thus obtain an instance $\widehat{\ell}: [N]^n \to \{\pm 1, \pm 2, \dots, \pm n\}$ of $nD$-\tucker. In particular, note that $\widehat{\ell}$ is antipodally anti-symmetric on the boundary, as required.\medskip

\noindent\textbf{Correctness.}
Any solution to the $nD$-\tucker instance $\widehat{\ell}$ yields $x,y \in K_m^n$ with $\|x-y\|_\infty \leq 1/m$ and $\ell(x) = -\ell(y)$. Without loss of generality, assume that $\ell(x) = +i$ for some $i \in [n]$. Since $\|x-y\|_\infty \leq 1/m$, we obtain that
$$\lambda(\mathcal{I}^+(x) \triangle~ \mathcal{I}^+(y)) \leq \sum_{j=1}^n |x_j-y_j| = \|x-y\|_1 \leq n \|x-y\|_\infty \leq n/m$$
and the same bound also holds for $\lambda(\mathcal{I}^-(x) \triangle~ \mathcal{I}^-(y))$. Since $v_i$ is Lipschitz-continuous with parameter $L$, it follows that
$$|v_i(\mathcal{I}^+(x)) - v_i(\mathcal{I}^+(y))| \leq L \cdot \lambda(\mathcal{I}^+(x) \triangle~ \mathcal{I}^+(y)) \leq nL/m$$
and similarly for $|v_i(\mathcal{I}^-(x)) - v_i(\mathcal{I}^-(y))|$.

Since $\ell(x) = +i$, it follows that $v_i(\mathcal{I}^+(x)) \geq v_i(\mathcal{I}^-(x))$. For the sake of contradiction, let us assume that $v_i(\mathcal{I}^+(x)) > v_i(\mathcal{I}^-(x)) + \varepsilon$. Then, it follows that
$$v_i(\mathcal{I}^+(y)) - v_i(\mathcal{I}^-(y)) \geq v_i(\mathcal{I}^+(x)) - v_i(\mathcal{I}^-(x)) - 2nL/m > \varepsilon - 2nL/m \geq 0$$
since $m \geq 2nL/\varepsilon$. But this contradicts the fact that $\ell(y) = -i$. Thus, it must hold that $|v_i(\mathcal{I}^+(x)) - v_i(\mathcal{I}^-(x))| \leq \varepsilon$. Since $\ell(x)=+i$, it follows that for all $j \in [n]$
$$|v_j(\mathcal{I}^+(x)) - v_j(\mathcal{I}^-(x))| \leq |v_i(\mathcal{I}^+(x)) - v_i(\mathcal{I}^-(x))| \leq \varepsilon.$$
This means that $\mathcal{I}^+(x),\mathcal{I}^-(x)$ yields a solution to the original \e-\ch\ instance.

Note that the reduction uses $N=2m+1 \leq 4nL/\varepsilon +3$ for the $nD$-\tucker instance. Furthermore, any query to the labelling function $\widehat{\ell}$ can be answered by performing $2n$ queries to the valuation functions $v_1, \dots, v_n$.


\section{General Valuations}\label{sec:general}

We are now ready to prove our main results for the $\e$-\ch\ problem, starting from the case of general valuations. First, for a single agent with a general valuation function, a simple binary search procedure is sufficient to solve $\e$-\ch\ with a polynomial number of queries and in polynomial time, therefore obtaining an efficient algorithm both in the white-box and in the black-box model. We have the following theorem.

\begin{theorem}\label{thm:CH-general-1agent}
For one agent with a general valuation function (or multiple agents with identical general valuations), $\e$-\ch\ is solvable in polynomial time and has query complexity $\Theta\left(\log \frac{L}{\e}\right)$.
\end{theorem}

\begin{proof}
We will prove the theorem for the case of $n=1$, as a solution to \e-\ch\ for this case is also straightforwardly a solution to the problem with multiple agents with identical valuations. 
Our algorithm essentially simulates binary search.
We say that the {\em label of a cut} $x \in [0,1]$ is \lplus, if $v([0,x]) > v([x,1]) + \e$. Respectively, the label of cut $x$ is \lminus, if $v([x,1]) > v([0,x]) + \e$. In any other case, the label of the cut is 0; if a cut has label 0, then it is a solution to \e-\ch. Observe that in order for an interval $[a,b] \subseteq [0,1]$ to contain a solution, it suffices that the label of $a$ be \lminus and the label of $b$ be \lplus (or vice-versa); then there is definitely a point $x \in [a,b]$ where the label is 0 (by continuity of $v$).

Now let $a=0$ and $b=1$. If $a$ or $b$ has label 0, then we have immediately found a solution. Otherwise, note that if $a$ has label \lminus, then $b$ must have label \lplus, and vice-versa. For convenience, in what follows, we assume that $a$ has label \lminus and $b$ has label \lplus. Our algorithm proceeds as follows in every iteration. Given an interval $[a,b]$ with label \lminus for $a$ and label \lplus for $b$, it computes the label of $\frac{a+b}{2}$. This can be done via two \eval queries. Then, if the label of $\frac{a+b}{2}$ is \lplus, it sets $b=\frac{a+b}{2}$; if the label is \lminus, it sets $a=\frac{a+b}{2}$; and if the label is 0 it outputs this cut.

We claim that the algorithm will always find a cut with label 0 after at most $\log \frac{L}{\e}$ iterations. For the sake of contradiction, assume that there is no such cut after $\log \frac{L}{\e}$ iterations. Observe that the length of $[a,b]$ in this case will be $\frac{\e}{L}$. In addition, we know the labels of $a$ and $b$. Cut $a$ has label \lminus, thus $v([a,1]) > v([0,a])+\e$, and cut $b$ has label \lplus, i.e., $v([0,b]) > v([b,1]) + \e$. Since $|b-a| \leq \frac{\e}{L}$ and $v$ is $L$-Lipschitz-continuous, it follows that
$$|v([a,1])-v([b,1])| \leq L \cdot \lambda([a,b)) \leq L \cdot \frac{\e}{L} = \e$$
and similarly $|v([0,a])-v([0,b])| \leq \e$. Putting everything together, we obtain that
$$v([0,b]) \leq v([0,a]) + \e < v([a,1]) \leq v([b,1]) + \e$$
which contradicts the assumption that cut $b$ has label \lplus.

Since a polynomial-time algorithm which queries the polynomial-time algorithm of the input $O\left(\log \frac{L}{\e}\right)$ times is a polynomial-time algorithm, we immediately obtain the polynomial-time upper bound for the white-box model.

For the black-box model, the algorithm immediately gives us the upper bound, whereas the lower bound follows from our general reduction from $nD$-\bu (\cref{prop:BU2CH}), and the query lower bounds for the latter problem obtained through \cref{lem:TUtoBU-general} below. In more detail, $1D$-\bu (and thus $\e$-\ch\ with a single agent) inherits its query complexity lower bounds from $1D$-\tucker, which can be easily seen to require at least $\Omega(\log N)$ queries in the worst-case. The latter bound naturally translates to a $\Omega(\log (L/\e))$ bound for $\e$-\ch. We also remark that the upper bound holds for any version of the problem with general valuations, even in the weak black-box model, whereas the lower bound holds even for normalised general valuations and for the standard black-box model.
\end{proof}

We now move to our results for two or more agents with general valuations. Here we obtain a $\ppa$-completeness result for $\e$-\ch, as well as exponential bounds on the query complexity of the problem. Our results demonstrate that for general valuations, even in the case of two agents, the problem is intractable in both the black-box and the white-box model. The main technical result of the section is the following pivotal lemma, proved at the end of this section.

\begin{lemma}\label{lem:TUtoBU-general}
For any constant $n \geq 1$, $nD$-\tucker reduces to normalised $nD$-\bu, via an efficient black-box reduction.
\end{lemma}

Now we state our main theorem about the computational/query complexity of the $\e$-\ch\ problem, as well as a corresponding theorem for $nD$-\bu. The proofs follow from \cref{lem:tuckerppa,lem:tucker-query} characterising the complexity of $nD$-\tucker and the following chain of reductions (where ``$\leq$'' denotes an efficient black-box reduction from the problem on the left-hand side to the problem on the right-hand side). 

\begin{mdframed}[backgroundcolor=white!90!gray,
      leftmargin=\dimexpr\leftmargin-20pt\relax,
      innerleftmargin=4pt,
      innertopmargin=0pt,
      skipabove=5pt,skipbelow=5pt]
\begin{align*}
\begin{small}
nD\text{-}\tucker  \underset{\cshref{lem:TUtoBU-general}}{\leq} nD\text{-}\bu \underset{\cshref{prop:BU2CH}}{\leq} \e\text{-}\ch \underset{\cshref{prop:CH2TU}}{\leq} nD\text{-}\tucker
\end{small}
\end{align*}
\end{mdframed}\vskip 5pt

\noindent The specific parameters that appear in the bounds below follow from the proof of \cref{lem:TUtoBU-general}.

\begin{theorem}\label{thm:CH-general-2agents}
Let $n \geq 2$ be any constant. Then,
\begin{itemize}
\item[-] $\e$-\ch\ with $n$ normalised general agents is \ppa-complete. This remains the case, even if (a) we fix $\varepsilon \in (0,1)$, or (b) we fix $L \geq 3(n+1)$;\smallskip
\item[-] there exists a constant $c > 0$ such that for any $\varepsilon \in (0,1)$ and any $L \geq 3(n+1)$ with $L/\varepsilon \geq c$, the query complexity of $\e$-\ch\ with $n$ normalised general agents is $\Theta ((L/\varepsilon)^{n-1})$.
\end{itemize}
\end{theorem}

\smallskip

\begin{theorem}\label{thm:BU-2agents}
Let $n \geq 2$ be any constant. Then,
\begin{itemize}
\item[-] normalised $nD$-\bu is \ppa-complete. This remains the case, even if (a) we fix $\varepsilon \in (0,1)$, or (b) we fix $L \geq 3$; \smallskip
\item[-] there exists a constant $c > 0$ such that for any $\varepsilon \in (0,1)$ and any $L \geq 3$ with $L/\varepsilon \geq c$, the query complexity of normalised $nD$-\bu is $\Theta ((L/\varepsilon)^{n-1})$.
\end{itemize}
\end{theorem}

\noindent In both cases, the lower bounds hold even for the normalised versions of the problems, while the upper bounds hold even for the more general, non-normalised, versions.

\subsection{Reducing \texorpdfstring{$\boldsymbol{nD}$}{nD}-Tucker to normalised \texorpdfstring{$\boldsymbol{nD}$}{nD}-Borsuk Ulam (Proof of \cref{lem:TUtoBU-general})}\label{sec:TUtoBU-general}

Let $n \geq 1$ be any constant. Consider an instance $\ell: [N]^n \to \{\pm 1, \pm 2, \dots, \pm n\}$ of $nD$-\tucker. Let $\varepsilon \in (0,1)$. We will construct a normalised $nD$-\bu function $F: B^{n+1} \rightarrow B^n$ that is Lipschitz-continuous with Lipschitz parameter $L=\max\{3,4n^2(N-1)\varepsilon+1\}$ and such that any $x \in \partial(B^{n+1})$ with $\|F(x)\|_\infty \leq \varepsilon$ yields a solution to the $nD$-\tucker instance.

Let $\delta = \min\{2\varepsilon,1\}$. Note that $\delta \in (0,1]$ and $\varepsilon < \delta < 2\varepsilon$. Without loss of generality, we can assume that for $p=(N,N,\dots,N)$ it holds $\ell(p) = +1$. Indeed, it is easy to see that we can rename the labels to achieve this, without introducing any new solutions.

The remainder of the proof will proceed in three steps: In Step 1, we will interpolate the $nD$-\tucker instance, to obtain a continuous function on $[-1/2,1/2]^n$, in Step 2, we extend this function to the whole domain $[-1,1]^n$, and in Step 3 we further extend to $[-1,1]^{n+1}$, to obtain an instance of the normalised $nD$-\bu problem.\medskip

\noindent\textbf{Step 1: Interpolating the $\boldsymbol{nD}$-\tucker instance.}
The first step is to embed the $nD$-\tucker grid $[N]^n$ in $[-1/2,1/2]^n$, define the value of the function at every grid point according to the labelling function $\ell$ and then interpolate to obtain a continuous function $f: [-1/2,1/2]^n \to [-\delta \cdot n,\delta \cdot n]$.

We embed the grid $[N]^n$ in $[-1/2,1/2]^n$ in a straightforward way, namely $p \in [N]^n$ corresponds to $\widehat{p} \in [-1/2,1/2]^n$ such that $\widehat{p}_j = -1/2 + (p_j-1)/(N-1)$ for all $j \in [n]$. Note that antipodal grid points exactly correspond to antipodal points in $[-1/2,1/2]^n$. In other words, $p$ and $q$ are antipodal on the grid, if and only if $\widehat{p} = - \widehat{q}$.

Next we define the value of the function $f: [-1/2,1/2]^n \to [-\delta \cdot n,\delta \cdot n]$ at the embedded grid points as follows
$$f(\widehat{p}) = \delta \cdot n \cdot e_{\ell(p)}$$
for all $p \in [N]^n$. For $i \in [n]$, $e_{+i}$ denotes the $i$th unit vector in $\mathbb{R}^n$, and $e_{-i} := - e_{+i}$. We then use Kuhn's triangulation on the embedded grid to interpolate between these values and obtain a function $f: [-1/2,1/2]^n \to [-\delta \cdot n,\delta \cdot n]$ (see \cref{sec:kuhn} for more details). We obtain:
\begin{itemize}
    \item $f$ is antipodally anti-symmetric on the boundary of $[-1/2,1/2]^n$, i.e., $f(-x) = -f(x)$ for all $x \in \partial([-1/2,1/2]^n)$.
    \item $f$ is Lipschitz-continuous with Lipschitz parameter $2n^2(N-1)\delta$, since the grid size is $1/(N-1)$ and $\|f(\widehat{p})\|_\infty \leq \delta \cdot n$ for all $p \in [N]^n$.
    \item Any $x \in [-1/2,1/2]^n$ such that $\|f(x)\|_\infty \leq \varepsilon$ must lie in a Kuhn simplex that contains two grid points $p,q \in [N]^n$ such that $\ell(p) = - \ell(q)$, i.e., a solution to the $nD$-\tucker instance. Indeed, let $p^0,p^1, \dots, p^n \in [N]^n$ be the grid points of the Kuhn simplex containing $x$. If $\{\ell(p^0),\ell(p^1),\dots,\ell(p^n)\}$ does not contain two opposite labels, then all the points in $V=\{f(\widehat{p^0}),f(\widehat{p^1}),\dots,f(\widehat{p^n})\}$ lie in the same orthant of $\mathbb{R}^n$. Since $\|f(\widehat{p^i})\|_\infty = \delta \cdot n$ for all $i$, it follows that any convex combination $v$ of vectors in $V$ must be such that $\|v\|_1 \geq \delta \cdot n$, and thus $\|v\|_\infty \geq \delta$. As a result, if $\|f(x)\|_\infty \leq \varepsilon < \delta$, then $\{\ell(p^0),\ell(p^1),\dots,\ell(p^n)\}$ must contain two opposite labels.
    \item $f(1/2,1/2,\dots,1/2) = \delta \cdot n \cdot e_{+1} = (\delta \cdot n,0,0,\dots,0)$.
\end{itemize}

Now, we define $g: [-1/2,1/2]^n \to [-\delta,\delta]^n$ to be the truncation of $f$ to $[-\delta,\delta]^n$, namely
$$g_i(x) = \min\{\delta, \max \{-\delta, f_i(x)\}\}.$$
It is not hard to see that $g$ is also antipodally anti-symmetric, Lipschitz-continuous with Lipschitz parameter $2n^2(N-1)\delta$ and $g(1/2,1/2,\dots,1/2) = \delta \cdot e_{+1}$. Furthermore, if $x \in [-1/2,1/2]^n$ is such that $\|g(x)\|_\infty \leq \varepsilon$, then, since $\varepsilon < \delta$, $\|f(x)\|_\infty \leq \varepsilon$, and thus $x$ again yields a solution to the $nD$-\tucker instance.\medskip

\noindent\textbf{Step 2: Extending to $\boldsymbol{[-1,1]^n}$.}
The goal of the next step is to define a function $h: [-1,1]^n \to [-1,1]^n$ that extends $g$ and ensures that $h(1,1,\dots,1)=(1,1,\dots,1)$, while maintaining its other properties. For $x \in [-1,1]^n$ we let $T(x) \in [-1/2,1/2]^n$ denote its truncation to $[-1/2,1/2]^n$, i.e., $[T(x)]_i = \min\{1/2, \max \{-1/2, x_i\}\}$ for all $i \in [n]$. The function $h: [-1,1]^n \to [-1,1]^n$ is defined as
\begin{equation*}
h(x) = \left\{\begin{tabular}{ll}
    $(2\min_j x_j - 1) \cdot \mathbb{1} + (2 - 2\min_j x_j) \cdot \delta \cdot e_{+1}$ & if $x_i \geq 1/2$ for all $i$ \\
    $(2\min_j (-x_j) - 1) \cdot (-\mathbb{1}) + (2 - 2\min_j (-x_j)) \cdot \delta \cdot e_{-1}$ & if $x_i \leq -1/2$ for all $i$ \\
    $g(T(x))$ & otherwise
\end{tabular}\right.
\end{equation*}
where $\mathbb{1} \in \mathbb{R}^n$ denotes the all-ones vector, i.e., $\mathbb{1}=(1,1,\dots,1)$. Clearly, it holds that $h(1,1,\dots,1) = \mathbb{1}$. It is also easy to see that $h(-x) = -h(x)$ for all $x \in \partial ([-1,1]^n)$, in particular because $T(-x)=-T(x)$. Furthermore, if $\|h(x)\|_\infty \leq \varepsilon$, then it must be that $\|g(T(x))\|_\infty \leq \varepsilon$, which yields a solution to the $nD$-\tucker instance. Indeed, if $x_i \geq 1/2$ for all $i$, then $$h_1(x) = (2\min_j x_j - 1) \cdot 1 + (2 - 2\min_j x_j) \cdot \delta \geq \delta > \varepsilon$$ so $\|h(x)\|_\infty > \varepsilon$. By the same argument, if $x_i < -1/2$ for all $i$, then we also have $\|h(x)\|_\infty > \varepsilon$.

Since $g(1/2,1/2,\dots,1/2) = \delta \cdot e_{+1}$ and $g(-1/2,-1/2,\dots,-1/2) = \delta \cdot e_{-1}$, it is easy to see that $h$ is continuous. Furthermore, since for any $x,y \in [-1,1]^n$ it holds that $\|T(x)-T(y)\|_\infty \leq \|x-y\|_\infty$, it is easy to see that $h$ is $2n^2(N-1)\delta$-Lipschitz-continuous outside of $\{x \in [-1,1]^n | x_i \geq 1/2 \text{ for all } i \in [n]\} \cup \{x \in [-1,1]^n | x_i \leq -1/2 \text{ for all } i \in [n]\}$. For any $y,z \in \{x \in [-1,1]^n | x_i \geq 1/2 \text{ for all } i \in [n]\}$, it holds that $|h_i(y)-h_i(z)| = 2|\min_j y_j - \min_j z_j| \leq 2\|y-z\|_\infty$ for $i > 1$, and $|h_1(y)-h_1(z)| = 2 (1-\delta) |\min_j y_j - \min_j z_j| \leq 2\|y-z\|_\infty$. Thus, $h$ is $2$-Lipschitz-continuous on $\{x \in [-1,1]^n | x_i \geq 1/2 \text{ for all } i \in [n]\}$ and, by the same argument, also on $\{x \in [-1,1]^n | x_i \leq -1/2 \text{ for all } i \in [n]\}$.

As a result, $h$ is Lipschitz-continuous on $[-1,1]^n$ with Lipschitz parameter $\max\{2,2n^2(N-1)\delta\}$. Indeed, consider any $x,y \in [-1,1]^n$. If $x_i \geq 1/2$ and $y_i \leq -1/2$ for all $i$, then $\|x-y\|_\infty \geq 1$, and thus $\|h(x) - h(y)\|_\infty \leq 2 \leq 2 \|x-y\|_\infty$. By symmetry, the only remaining case that we need to check is when $x_i \geq 1/2$ for all $i$, and $y$ is such that there exists $i$ with $y_i < 1/2$ and there exists $i$ with $y_i > -1/2$. In that case, we consider the segment $[x,y]$ from $x$ to $y$, and let $z \in [x,y]$ be the point that is the furthest away from $x$ but such that $z_i \geq 1/2$ for all $i$. Note that there must exist $i$ such that $z_i=1/2$. This means $h(z)=g(T(z))$ and thus $\|h(z)-h(y)\|_\infty \leq 2n^2(N-1)\delta \|z-y\|_\infty$. On the other hand, we have $\|h(x)-h(z)\|_\infty \leq 2\|x-z\|_\infty$. Putting these two expressions together, we obtain that $\|h(x)-h(y)\|_\infty \leq \max\{2,2n^2(N-1)\delta\} (\|x-z\|_\infty + \|z-y\|_\infty) = \max\{2,2n^2(N-1)\delta\} \|x-y\|_\infty$. Here we used the fact that $z \in [x,y]$, which means that there exists $t \in [0,1]$ such that $z = x+t(y-x)$ and thus
$$\|x-z\|_\infty + \|z-y\|_\infty = t \|x-y\|_\infty + (1-t) \|x-y\|_\infty = \|x-y\|_\infty.$$\medskip

\noindent\textbf{Step 3: Extending to $\boldsymbol{[-1,1]^{n+1}}$.}
The final step is to define a normalised $nD$-\bu function $F: [-1,1]^{n+1} \to [-1,1]^n$ such that any $x \in \partial([-1,1]^{n+1})$ with $\|F(x)\|_\infty \leq \varepsilon$ yields a solution to the $nD$-\tucker instance. For $x \in [-1,1]^{n+1}$ we write $x=(x',x_{n+1})$, where $x' \in [-1,1]^n$. We define
$$F(x) = F(x',x_{n+1}) = \frac{1+x_{n+1}}{2} h(x') + \frac{1-x_{n+1}}{2} (-h(-x')).$$
Since $h(x'),-h(-x') \in [-1,1]$ and $F(x)$ is a convex combination of these two, it follows that $F(x) \in [-1,1]^n$. Furthermore, we have $F(1,1,\dots,1) = h(1,1,\dots,1) = \mathbb{1}$. $F$ is an odd function, since
$$F(-x) = F(-x',-x_{n+1}) = \frac{1-x_{n+1}}{2} h(-x') + \frac{1+x_{n+1}}{2} (-h(x')) = -F(x',x_{n+1}) = -F(x).$$
Consider any $x = (x',x_{n+1}) \in \partial([-1,1]^{n+1})$ with $\|F(x)\|_\infty \leq \varepsilon$. Since $F$ is an odd function, we can assume that $x_{n+1} \geq 0$ (otherwise just use $-x$ instead of $x$). If $x_{n+1}=1$, then $F(x',x_{n+1}) = h(x')$, and thus $\|h(x')\|_\infty \leq \varepsilon$, which yields a solution to the $nD$-\tucker instance. If $x_{n+1} \in [0,1)$, then $x' \in \partial([-1,1]^n)$ and thus $h(x') = -h(-x')$. This implies that $F(x',x_{n+1}) = h(x')$ in this case too.

Finally, let us determine the Lipschitz parameter of $F$. Let $x,y \in [-1,1]^{n+1}$. We have
\begin{equation*}
\begin{split}
\|F(x',x_{n+1}) - F(y',x_{n+1})\|_\infty &\leq \frac{1+x_{n+1}}{2} \|h(x')-h(y')\|_\infty + \frac{1-x_{n+1}}{2} \|h(-x') - h(-y')\|_\infty\\
&\leq \max\{2,2n^2(N-1)\delta\} \|x'-y'\|_\infty
\end{split}
\end{equation*}
and also
$$\|F(y',x_{n+1}) - F(y',y_{n+1})\|_\infty \leq \frac{|x_{n+1}-y_{n+1}|}{2} (\|h(y')\|_\infty + \|h(-y')\|_\infty) \leq |x_{n+1}-y_{n+1}|.$$
Putting these two expressions together, it follows that
$$\|F(x) - F(y)\|_\infty \leq \max\{3,2n^2(N-1)\delta+1\} \|x-y\|_\infty.$$
Note that $\max\{3,2n^2(N-1)\delta+1\} \leq \max\{3,4n^2(N-1)\varepsilon+1\}$.\medskip

In the black-box model, in order to answer one query to $F$, we have to answer two queries to $h$, i.e., two queries to $g$. In order to answer a query to $g$, we have to answer one query to $f$, i.e., $n+1$ queries to the labelling function $\ell$ (in order to interpolate). Thus, one query to $F$ requires $2(n+1)$ queries to $\ell$. Since $n$ is a constant, the query lower bounds from $nD$-\tucker carry over to normalised $nD$-\bu.

In the white-box model, the reduction actually gives us a way to construct an arithmetic circuit that computes $F$, if we are given a Boolean circuit that computes $\ell$. Indeed, using standard techniques \citep{chen2009settling,Daskalakis2009}, the execution of the Boolean circuit on some input can be simulated by the arithmetic circuit. Furthermore, the input bits for the Boolean circuit can be obtained by using the $<$ gate. All the other operations that we used to construct $F$ can be computed by the arithmetic gates $+,-,\times, \max, \min, <$ and rational constants. Thus, we obtain a polynomial-time many-one reduction from $nD$-\tucker to normalised $nD$-\bu for all $n \geq 1$.


\section{Monotone Valuations}\label{sec:monotone}

In this section, we present our results for agents with monotone valuations. In contrast to the results of \cref{sec:general}, here we prove that for two agents with monotone valuations, the problem is solvable in polynomial time and with a polynomial number of queries, and in fact this result holds even if only one of the two agent has a monotone valuation and the other has a general valuation. For three or more agents however, the problem becomes \ppa-complete once again, and we obtain a corresponding exponential lower bound on its query complexity.

\subsection{An efficient algorithm for two monotone agents}

We start with our efficient algorithm for the case of two agents, which is a polynomial-time algorithm in the white-box model, as well as an algorithm of polynomial query complexity in the black-box model; see Algorithm~\ref{alg:two-agents-monotone}. The algorithm is based on a \emph{nested binary search} procedure. At the higher level, we are performing a binary search on the position of the left cut of a solution. At the lower level, for any fixed position for the left cut, we perform another binary search in order to find a right cut such that the pair of cuts forms a solution for the first agent; as we have already seen this can be efficiently done if the agent has monotone valuation. Intuitively, we are moving on the ``indifference curve'' of the valuation function of the agent with the monotone valuation (see the red zig-zag line in \cref{fig:monotone}) until we reach a solution. We decide how to move on this curve by checking the preferences of the second agent. \medskip

\noindent Before we proceed, we draw an interesting connection with \emph{Austin's moving knife procedure} \citep{austin1982sharing}, an Exact-Division procedure for two agents with general valuations. The procedure is based on two moving knifes which one of the two agents simultaneously and continuously slides across the cake, maintaining that the corresponding cut positions ensure that she is satisfied with the partition. At some point during this process, the other agent becomes satisfied, which is guaranteed by the intermediate value theorem. Our algorithm can be interpreted as a discrete time implementation of this idea and quite interestingly, it results in a polynomial-time algorithm when one of the two agents has a monotone valuation, whereas it is computationally hard when both agents have general valuations, as shown in \cref{sec:general}. On a more fundamental level, this demonstrates the intricacies of transforming moving-knife fair division protocols into discrete algorithms.\medskip

\noindent The main theorem of this section is the following.

\begin{theorem}\label{thm:CH-monotone-2agent}
For two agents with monotone valuation functions, $\e$-\ch\ is solvable in polynomial time and has query complexity $O\left(\log^2 \frac{L}{\e}\right)$, even in the weak black-box model. This result holds even if one of the two agents has a general valuation.
\end{theorem}

\begin{figure}
  \begin{center}
      \includegraphics[width=\linewidth]{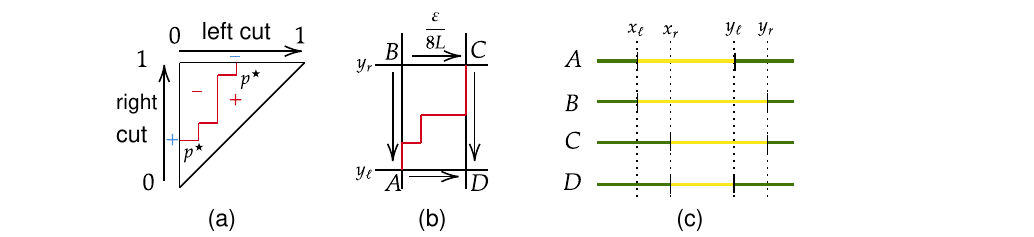}
  \end{center}
    \caption{Visualisation of Algorithm~\ref{alg:two-agents-monotone}. (a) depicts our initial assumptions. The red line shows where Agent 1 is indifferent. The blue signs on $(0,p^{\star})$ and $(p^{\star}, 1)$ show the (weak) preferences of Agent 2 under these pairs of cuts. (b) shows a possible position for the cuts $x_r - x_\ell \leq \frac{\e}{8L}$. The arrows show how the difference between the values of the positive piece and the negative piece change between the four possible combinations of pairs of cuts. (c) depicts the actual cuts on the cake: the green parts have label \lplus and the yellow parts have label \lminus.}
    \label{fig:monotone}
\end{figure}

Before we proceed with the description of our algorithm and its analysis, let us begin with some conventions that will make the presentation easier. Since we have to make two cuts, we denote \lcut the position of the leftmost cut and \rcut the position of the rightmost cut. So, $0 \leq \lcut \leq \rcut \leq 1$. In addition, the labels of the corresponding intervals are as follows: intervals $[0,\lcut]$ and $[\rcut, 1]$ have label \lplus, forming the positive piece, and interval $[\lcut, \rcut]$ has label \lminus, forming the negative piece. Given a pair of cuts $(\lcut, \rcut)$, we say that agent $i$:
\begin{itemize}
    \item weakly prefers the positive piece, if $v_i([0,\lcut] \cup [\rcut, 1]) \geq v_i([\lcut,\rcut]) - \e/2$;  
    \item weakly prefers the negative piece, if $v_i([\lcut,\rcut]) \geq v_i([0,\lcut] \cup [\rcut, 1]) - \e/2$; 
    \item is indifferent if $|v_i([\lcut,\rcut]) - v_i([0,\lcut] \cup [\rcut, 1])| \leq \e/2$.
\end{itemize}
In addition, let $p^{\star} \in [0,1]$ be such that $|v_1([0,p^{\star}]) - v_1([p^{\star},1])| \leq \e/2$. Note that we can efficiently compute $p^{\star}$ using \cref{thm:CH-general-1agent}.
The final assumption we need to make is regarding the preferences of Agent 2 for the two special pairs of $(0,p^{\star})$ and $(p^{\star}, 1)$.
Observe that both pairs of cuts create the same pieces over $[0,1]$ and only change the labels of the pieces. 
Hence, if Agent 2 weakly prefers the positive piece under the pair of cuts $(0,p^{\star})$, he {\em has to} weakly prefer the negative piece under the pair of cuts $(p^{\star},1)$. In the description of the algorithm we will assume that this is indeed the case, i.e., he weakly prefers the positive piece under $(0,p^{\star})$ and the negative piece under $(p^{\star},1)$. (The other case can be handled analogously.)
Using the above notation and assumptions we can now state Algorithm~\ref{alg:two-agents-monotone}.

\begin{algorithm}[h]
\SetAlgoLined
\caption{\e-\ch\ for two agents with monotone valuations}\label{alg:two-agents-monotone}
Set $x_\ell \leftarrow 0$ and $x_r \leftarrow p^{\star}$ \\
Set $y_\ell \leftarrow p^{\star}$ and $y_r \leftarrow 1$ \\
\While{$x_r-x_\ell > \frac{\e}{8L}$}{
Find $y \in [y_\ell,y_r]$ such that Agent 1 is indifferent under the pair of cuts $(\frac{x_\ell+x_r}{2},y)$ \label{step-w1}\\
\uIf{Agent 2 weakly prefers the positive piece under $(\frac{x_\ell+x_r}{2},y)$}{
Set $x_\ell \leftarrow \frac{x_\ell+x_r}{2}$ and $y_\ell \leftarrow y$
}
\ElseIf{Agent 2 weakly prefers the negative piece under $(\frac{x_\ell+x_r}{2},y)$}{
Set $x_r \leftarrow \frac{x_\ell+x_r}{2}$ and $y_r \leftarrow y$
}
}
\While{$y_r-y_\ell > \frac{\e}{8L}$}{
Find $x \in [x_\ell,x_r]$ such that Agent 1 is indifferent under the pair of cuts $(x,\frac{y_\ell+y_r}{2})$ \label{step-w2}\\
\uIf{Agent 2 weakly prefers the positive piece under $(x,\frac{y_\ell+y_r}{2})$}{
Set $y_\ell \leftarrow \frac{y_\ell+y_r}{2}$ and $x_\ell \leftarrow x$
}
\ElseIf{Agent 2 weakly prefers the negative piece under $(x,\frac{y_\ell+y_r}{2})$}{
Set $y_r \leftarrow \frac{y_\ell+y_r}{2}$ and $x_r \leftarrow x$
}
}
Output $(x_\ell,y_\ell)$
\end{algorithm}

\begin{proof}[Proof of \cref{thm:CH-monotone-2agent}]
To prove the correctness of Algorithm~\ref{alg:two-agents-monotone} we will prove that the following invariants are maintained through all iterations of the algorithm.
\begin{enumerate}
    \item Agent 1 is indifferent under the pair of cuts $(x_\ell,y_\ell)$, and also under the pair of cuts $(x_r,y_r)$.
    \item Agent 2 weakly prefers the positive piece under the pair of cuts $(x_\ell,y_\ell)$, and weakly prefers the negative piece under the pair of cuts $(x_r,y_r)$.
\end{enumerate}

Assuming that Invariants 1 and 2 hold, it follows that the algorithm outputs a correct solution. Indeed, Agent 1 is indifferent under the pair of cuts $(x_\ell,y_\ell)$ by Invariant $1$. By Invariant 2, Agent 2 weakly prefers the positive piece under the pair of cuts $(x_\ell,y_\ell)$, i.e., $v_2([0,x_\ell] \cup [y_\ell, 1]) \geq v_2([x_\ell,y_\ell]) - \e/2 \geq v_2([x_\ell,y_\ell]) - \e$. Thus, it suffices to show that $v_2([x_\ell,y_\ell]) \geq v_2([0,x_\ell] \cup [y_\ell, 1]) - \e$. By Invariant 2, Agent 2 weakly prefers the negative piece under the pair of cuts $(x_r,y_r)$, i.e., $v_2([0,x_r] \cup [y_r, 1]) \geq v_2([x_r,y_r]) - \e/2$. Since $|x_r-x_\ell| \leq \e/8L$ and $|y_r-y_\ell| \leq \e/8L$, by the $L$-Lipschitz continuity of $v_2$ it follows that $|v_2([0,x_r] \cup [y_r, 1]) - v_2([0,x_\ell] \cup [y_\ell, 1])| \leq \e/4$ and $|v_2([x_r,y_r]) - v_2([x_\ell,y_\ell])| \leq \e/4$. As a result, $v_2([0,x_\ell] \cup [y_\ell, 1]) \geq v_2([x_\ell,y_\ell]) - \e$.

Next, we prove that Invariants 1 and 2 hold. First of all, note that they hold at the start of the algorithm by the choice of $p^{\star}$ and from the fact that Agent 2 weakly prefers the positive piece under $(0,p^{\star})$ and the negative piece under $(p^{\star},1)$. Both invariants are then automatically maintained by construction of the algorithm. We just have to argue that Steps~\ref{step-w1} and~\ref{step-w2} are possible, i.e., that such $y$ (respectively $x$) exists. Consider Step~\ref{step-w1} first. We apply the intermediate value theorem as follows: for the pair of cuts $(\frac{x_\ell+x_r}{2},y)$, Agent 1 weakly prefers the positive piece when $y = y_\ell$, and weakly prefers the negative piece when $y = y_r$. Thus, by continuity of the valuation, there exists $y \in [y_\ell,y_r]$ such that Agent 1 is indifferent between the two pieces. For the first point, note that Agent 1 is indifferent for the pair of cuts $(x_\ell,y_\ell)$ (by Invariant 1 in the previous iteration), and thus weakly prefers the positive piece for the pair of cuts $(\frac{x_\ell+x_r}{2},y_\ell)$ by monotonicity, since the positive piece has increased. For the second point, note that Agent 1 is indifferent for the pair of cuts $(x_r,y_r)$ (again by Invariant 1 in the previous iteration), and thus weakly prefers the negative piece for the pair of cuts $(\frac{x_\ell+x_r}{2},y_r)$ by monotonicity, since the negative piece has increased. The same argument also applies to Step~\ref{step-w2}. Finally, observe that we have not made use of the monotonicity of Agent 2's valuation anywhere in our proof. Thus, the algorithm also works if Agent 2 has a general valuation.

In order to bound the running time of Algorithm~\ref{alg:two-agents-monotone}, we need to bound the running time of: the number of iterations the algorithm needs  in each while loop, and the time we need to perform Step~\ref{step-w1} and Step~\ref{step-w2}. 
Firstly, observe that Steps~\ref{step-w1} and~\ref{step-w2} can be tackled using the algorithm from \cref{thm:CH-general-1agent}. This is because we can view the problem as a special case where we have one agent and we need to place a single cut in a specific subinterval where we know that a solution exists. Thus, each of these steps requires $O(\log \frac{L}{\e})$ time. In addition, observe in every while loop we essentially perform a binary search. Thus, after $O(\log\frac{L}{\e})$ iterations we get that $x_r-x_\ell \leq \frac{\e}{8L}$, since in every iteration the distance between $x_r$ and $x_\ell$ decreases by a factor of 2. Similarly, after $O(\log\frac{L}{\e})$ iterations we get that $|y_r-y_\ell| \leq \frac{\e}{8L}$ as well. Hence, every while loop requires $O(\log^2\frac{L}{\e})$ time. So, Algorithm~\ref{alg:two-agents-monotone} terminates in $O(\log^2\frac{L}{\e})$ time. Finally, observe that our algorithm just requires the preferences of Agent 2 for specific pairs of cuts which can be done via two evaluations of $v_2$.

Next we argue that we can implement Algorithm~\ref{alg:two-agents-monotone} using $O(\log^2\frac{L}{\e})$ queries in the black-box model. Observe that Steps~\ref{step-w1} and~\ref{step-w2} can be simulated with $O(\log \frac{L}{\e})$ queries each, as we have already explained in \cref{thm:CH-general-1agent}. In addition, observe that every time we ask Agent 2 for his preferences we only need two evaluation queries.  Since we need $O(\log \frac{L}{\e})$ iterations in every while loop, Algorithm~\ref{alg:two-agents-monotone} needs $O(\log^2\frac{L}{\e})$ queries in total. 
\end{proof}

\subsection{Results for three or more monotone agents}

\noindent We now move on to the case of three or more monotone agents, for which we manage to show that the problem becomes computationally hard and has exponential query complexity. Our results thus show a clear dichotomy on the complexity of $\e$-\ch\ with monotone agents, between the case of two agents and the case of three or more agents.

Again we employ our general approach, but this time we need to prove computational and query-complexity hardness of the \emph{monotone} $nD$-\bu problem; the corresponding impossibility results for $\e$-\ch\ with agents with monotone valuations then follow from our property-preserving reduction (\cref{prop:BU2CH}). To this end, we in fact construct an efficient black-box reduction from $(n-1)D$-\tucker to monotone $nD$-\bu, i.e., we reduce from the corresponding version of $nD$-\tucker of one lower dimension. In order to achieve this, we once again interpolate the $(n-1)D$-\tucker instance to obtain a continuous function, but, this time, we embed it in a very specific lower dimensional subset of the $nD$-\bu domain. We then show that the function can be extended to a monotone function on the whole domain.\\

\noindent The ``drop in dimension'' which is featured in our reduction has the following effects:
\begin{itemize}
    \item[-] Since $1D$-\tucker is solvable in polynomial-time, we can only obtain the \ppa-hardness of monotone $nD$-\bu for $n \geq 3$, and therefore the \ppa-hardness of $\e$-\ch\ for three or more monotone agents.
    \item[-] The query complexity lower bounds that we ``inherit'' from $(n-1)D$-\tucker do not \emph{exactly} match our upper bounds, obtained via the reduction from $\e$-\ch\ to $nD$-\tucker (\cref{prop:CH2TU}).
\end{itemize}

\noindent The main technical contribution of this section is the following lemma, proved at the end of this section.

\begin{lemma}\label{lem:TUtoBU-monotone}
For any constant $n \geq 2$, $(n-1)D$-\tucker reduces to normalised monotone $nD$-\bu in polynomial time, via an efficient black-box reduction.
\end{lemma}

\noindent Similarly to \cref{sec:general}, we then obtain the following two theorems. 

\begin{theorem}\label{thm:CH-monotone-3agents}
Let $n \geq 3$ be any constant. Then,
\begin{itemize}
\item[-] $\e$-\ch\ with $n$ monotone agents is \ppa-complete. This remains the case, even if (a) we fix $\varepsilon \in (0,1)$, or (b) we fix $L \geq 3(n+1)$;\smallskip
\item[-] for any constant $t \in (0,1)$, there exists a constant $c > 0$ such that for any $\varepsilon \in (0,t)$ and any $L \geq 3(n+1)$ with $L/\varepsilon \geq c$, the query complexity of $\e$-\ch\ with $n$ monotone agents is between $\Omega ((L/\varepsilon)^{n-2})$ and $O((L/\varepsilon)^{n-1})$.
\end{itemize}
\end{theorem}

\begin{theorem}\label{thm:BU-monotone}
Let $n \geq 3$ be any constant. Then,
\begin{itemize}
\item[-] monotone $nD$-\bu is \ppa-complete. This remains the case, even if (a) we fix $\varepsilon \in (0,1)$, or (b) we fix $L \geq 3$; \smallskip
\item[-] for any constant $t \in (0,1)$, there exists a constant $c > 0$ such that for any $\varepsilon \in (0,t)$ and any $L \geq 3$ with $L/\varepsilon \geq c$, the query complexity of monotone $nD$-\bu is between $\Omega ((L/\varepsilon)^{n-2})$ and $O ((L/\varepsilon)^{n-1})$.
\end{itemize}
\end{theorem}

The proofs of the theorems follow from \cref{lem:tuckerppa,lem:tucker-query} and the following chain of reductions, which now crucially involve the monotone version of $nD$-\bu\ and $\e$-\ch.

\vskip 5pt
\begin{mdframed}[backgroundcolor=white!90!gray,
      leftmargin=\dimexpr\leftmargin-20pt\relax,
      innerleftmargin=4pt,
      innertopmargin=0pt,
      skipabove=5pt,skipbelow=5pt]
\begin{align*}
\begin{small}
(n-1)D\text{-}\tucker \underset{\strut\mathclap{\cshref{lem:TUtoBU-monotone}}}{\leq}  \text{\emph{monotone} }nD\text{-}\bu \underset{\strut\mathclap{\cshref{prop:BU2CH}}}{\leq} \text{\emph{monotone} }\e\text{-}\ch \underset{\strut\mathclap{\cshref{prop:CH2TU}}}{\leq} nD\text{-}\tucker
\end{small}
\end{align*}
\end{mdframed}\vskip 5pt

\noindent The specific parameters that appear in the bounds follow from the proof of \cref{lem:TUtoBU-monotone}. The lower bounds again hold even for the normalised versions of the problems. There is a small gap between our lower and upper bounds; we conjecture that it should be possible to prove upper bounds that match the lower bounds of \cref{thm:CH-monotone-3agents}, at least up to logarithmic factors, but we leave this for future work.

\subsection{Reducing \texorpdfstring{$\boldsymbol{(n-1)D}$}{(n-1)D}-\tucker to monotone \texorpdfstring{$\boldsymbol{nD}$}{nD}-\bu (Proof of \cref{lem:TUtoBU-monotone})} \label{sec:TUtoBU-monotone}

Let $n \geq 2$ be any fixed constant. Consider an instance $\ell: [N]^{n-1} \to \{\pm 1, \pm 2, \dots, \pm (n-1)\}$ of $(n-1)D$-\tucker. Let $\varepsilon \in (0,1)$. We will construct a monotone normalised $nD$-\bu function $F: [-1,1]^{n+1} \rightarrow [-1,1]^n$ that is Lipschitz-continuous with Lipschitz parameter $L=(4(n+1)^4N\varepsilon + 2)/(\min\{2,1/\varepsilon\}-1)$ and such that any $x \in \partial(B^{n+1})$ with $\|F(x)\|_\infty \leq \varepsilon$ yields a solution to the $(n-1)D$-\tucker instance.\medskip

\noindent\textbf{Step 1: From $\boldsymbol{(n-1)D}$-\tucker to non-normalised $\boldsymbol{(n-1)D}$-\bu.}
Let $\delta = \min \{2\varepsilon,1\}$. Note that $\delta \in (0,1]$ and $\varepsilon < \delta < 2\varepsilon$. The first step of the proof is very similar to the proof of \cref{lem:TUtoBU-general}. Namely, using Step 1 of the proof of \cref{lem:TUtoBU-general}, we construct $g: [-1/2,1/2]^{n-1} \to [-\delta,\delta]^{n-1}$ such that $g$ is antipodally anti-symmetric and $4(n-1)^2(N-1)\varepsilon$-Lipschitz-continuous. Furthermore, any $x \in [-1/2,1/2]^{n-1}$ with $\|g(x)\|_\infty < \delta$ yields a solution to the original $(n-1)D$-\tucker instance. Then, we define $h: [-1,1]^{n-1} \to [-\delta,\delta]^{n-1}$ by $h(x)=g(x/2)$. $h$ has the same properties as $g$, except that it is $2(n-1)^2(N-1)\varepsilon$-Lipschitz-continuous. Note that here the construction of $h$ differs from Step 2 of the proof of \cref{lem:TUtoBU-general}, since we do not want $h$ to be normalised.

Next, we extend $h$ to a (non-normalised) $(n-1)D$-\bu function $G: [-1,1]^n \to [-\delta,\delta]^{n-1}$ using the same construction as in Step 3 of the proof of \cref{lem:TUtoBU-general}. By the same arguments, it holds that $G$ is an odd function and it is Lipschitz-continuous with parameter $2(n-1)^2(N-1)\varepsilon + \delta \leq 2n^2N\varepsilon$. Furthermore, any $x \in \partial([-1,1]^n)$ with $\|G(x)\|_\infty < \delta$ yields a solution to the original $(n-1)D$-\tucker instance.

\bigskip

\noindent On a high-level, the rest of the reduction, which is the most interesting part, works by embedding $G$ in an $n$-dimensional subspace of $\mathbb{R}^{n+1}$ and then carefully extending it to all of $[-1,1]^{n+1}$ in a monotonic way. This $n$-dimensional subspace is
$$\mathcal{D} := \left\{x \in \mathbb{R}^{n+1} \left| \, \sum_{i=1}^{n+1} x_i = 0\right\}\right. .$$
It has the nice property that for any $x,y \in \mathcal{D}$, if $x \leq y$, then $x=y$.\medskip

\noindent\textbf{Step 2: Embedding into a function $\boldsymbol{\mathcal{D} \to [-\delta,\delta]^{n-1}}$.}
We begin by defining a slightly modified version of $G$. The function $\widehat{G}: \mathbb{R}^n \to [-\delta,\delta]^{n-1}$ is defined as follows
\begin{equation*}
\widehat{G}(x) = G\left((n+1) \cdot T_{\frac{1}{n+1}}(x)\right)
\end{equation*}
where $T_{r} : [-1,1]^n \to [-r,r]^n$ denotes truncation to $[-r,r]$ in every coordinate. It is easy to see that $\widehat{G}$ remains an odd function and that it is Lipschitz-continuous with parameter $(n+1) \cdot 2n^2N\varepsilon \leq 2(n+1)^3N\varepsilon$. Furthermore, it holds that any $x \in \mathbb{R}^n \setminus (-\frac{1}{n+1},\frac{1}{n+1})^n$ with $\|\widehat{G}(x)\|_\infty < \delta$ yields a solution to the $(n-1)D$-\bu instance $G$.

Next, we embed $\widehat{G}$ into $\mathcal{D}$. Let $H: \mathcal{D} \to [-\delta,\delta]^{n-1}$ be defined by $H(x) = H(x',x_{n+1}) = \widehat{G}(x')$. Note that $H$ remains an odd function and is also $2(n+1)^3N\varepsilon$-Lipschitz-continuous. Any $x \in \mathcal{D}$ with $\|H(x)\|_\infty < \delta$ and such that there exists $i \in [n]$ with $|x_i| \geq 1/(n+1)$, yields a solution to $\widehat{G}$ and thus to the original $(n-1)D$-\tucker instance.\medskip

\noindent\textbf{Step 3: Extending to a function $\boldsymbol{[-1,1]^{n+1} \to \mathbb{R}^{n-1}}$.}
In the next step, we extend $H$ to a function $F': [-1,1]^{n+1} \to \mathbb{R}^{n-1}$. For $x \in [-1,1]^{n+1}$, we let $S(x) = \sum_{i=1}^{n+1} x_i \in [-(n+1),n+1]$ and $\Pi(x) = x - \langle x, \mathbb{1}_{n+1} \rangle \cdot \mathbb{1}_{n+1} / (n+1) = x - S(x) \cdot \mathbb{1}_{n+1} /(n+1) \in \mathcal{D}$ denote the orthogonal projection onto $\mathcal{D}$. Here $\langle \cdot, \cdot \rangle$ denotes the scalar product in $\mathbb{R}^{n+1}$, and $\mathbb{1}_{n+1} \in \mathbb{R}^{n+1}$ is the all-ones vector. $F'$ is defined as follows:
\begin{equation*}
F'(x) = \left(1 - \frac{|S(x)|}{n+1}\right) \cdot H(\Pi(x)) + C \cdot \frac{S(x)}{n+1} \cdot \mathbb{1}_{n-1}
\end{equation*}
where $C = 1 + 2(n+1)^4N\varepsilon$. It is easy to check that $F'$ is an odd function by using the fact that $S(-x)=-S(x)$, $\Pi(-x)=-\Pi(x)$ and $H(-x) = -H(x)$.

It is easy to see that $F'$ is continuous. Let us determine an upper bound on its Lipschitz parameter. For any $x,y \in [-1,1]^{n+1}$ we have
\begin{equation*}
\begin{split}
\|F'(x) - F'(y)\|_\infty &\leq \frac{\big||S(x)|-|S(y)|\big|}{n+1} \|H(\Pi(x))\|_\infty + \big(1-|S(y)|/(n+1)\big) \|H(\Pi(x))-H(\Pi(y))\|_\infty\\
&\qquad+ C |S(x)-S(y)|/(n+1)\\
&\leq \|x-y\|_\infty + 2(n+1)^3N\varepsilon \sqrt{n+1} \|x-y\|_\infty + C \|x-y\|_\infty\\
&\leq (2(n+1)^4N\varepsilon+C+1) \|x-y\|_\infty
\end{split}
\end{equation*}
where we used $|S(x)-S(y)| \leq (n+1) \|x-y\|_\infty$ and $\|\Pi(x)-\Pi(y)\|_\infty \leq \|x-y\|_2 \leq \sqrt{n+1} \|x-y\|_\infty$. Thus, $F'$ is Lipschitz-continuous with Lipschitz parameter $2(n+1)^4N\varepsilon+C+1 = 4(n+1)^4N\varepsilon + 2$.

Let us now show that $F'$ is monotone. For this, it is enough to show that $F'(x) \leq F'(y)$ for all $x,y \in [-1,1]^{n+1}$ with $x \leq y$ and $S(x) \geq 0$, $S(y) \geq 0$. Indeed, since $F'$ is odd, this implies that the statement also holds if $S(x) \leq 0$ and $S(y) \leq 0$. Finally, if $S(x) \leq 0$ and $S(y) \geq 0$, then there exists $z$ with $x \leq z \leq y$ and $S(z)=0$, which implies that $F'(x) \leq F'(z) \leq F'(y)$.

Let $x \in [-1,1]^{n+1}$ be such that $S(x) \geq 0$. For any $j \in [n+1]$, $i \in [n-1]$ and $t \geq 0$ with $x_j+t \leq 1$, it holds that
\begin{equation*}
\begin{split}
&\qquad F_i'(x+t \cdot e_j) - F_i'(x)\\
&= -\frac{t}{n+1} H_i(\Pi(x+t\cdot e_j)) + \big(1-S(x)/(n+1)\big) \big(H_i(\Pi(x+t\cdot e_j)) - H_i(\Pi(x))\big) + \frac{tC}{n+1}\\
&\geq -\frac{t}{n+1} - 2(n+1)^3N\varepsilon \|\Pi(x+t\cdot e_j) - \Pi(x)\|_\infty + \frac{tC}{n+1}\\
&\geq (C - 1 - 2(n+1)^4N\varepsilon) \cdot \frac{t}{n+1} \geq 0
\end{split}
\end{equation*}
since $C = 1 + 2(n+1)^4N\varepsilon$.
Here we used the fact that $\|\Pi(x+t\cdot e_j) - \Pi(x)\|_\infty \leq \|\Pi(x+t\cdot e_j) - \Pi(x)\|_2 \leq \|x+t\cdot e_j - x\|_2 = t$.
We obtain that $F'$ is monotone, since for any $x,y$ with $x \leq y$ and $S(x) \geq 0$, $S(y) \geq 0$, we can decompose $y = (\dots((x + (y_1-x_1) \cdot e_1) + (y_2-x_2) \cdot e_2)\dots) + (y_{n+1}-x_{n+1}) \cdot e_{n+1}$.\medskip

\noindent\textbf{Step 4: Extending to a function $\boldsymbol{[-1,1]^{n+1} \to [-1,1]^n}$.}
Define $F'': [-1,1]^{n+1} \to \mathbb{R}$ by $F''(x) = C_n \cdot S(x)/(n+1)$, where $C_n := (4(n+1)^4N\varepsilon + 2)/(\min\{2,1/\varepsilon\}-1)$. Clearly, $F''$ is an odd function, monotone and $C_n$-Lipschitz-continuous. Finally, we define $F: [-1,1]^{n+1} \to [-1,1]^n$ by $F(x) = T_1((F'(x),F''(x)))$. It is easy to check that $F$ remains monotone and odd, because $T_1$, $F'$ and $F''$ are monotone and odd. Furthermore, since $C \geq 1$ and $C_n \geq 1$, we have that $F(1,1,\dots,1) = (1,1,\dots,1)$. Thus, $F$ is a monotone normalised $nD$-\bu function with Lipschitz parameter $\max\{4(n+1)^4N\varepsilon + 2,C_n\} = (4(n+1)^4N\varepsilon + 2)/(\min\{2,1/\varepsilon\}-1)$.

Consider any $x \in \partial([-1,1]^{n+1})$ with $\|F(x)\|_\infty \leq \varepsilon$. Since $|F''(x)| \leq \varepsilon$, it follows that $|S(x)| \leq (n+1) \varepsilon /C_n$. Assume that there exists $i \in [n-1]$ such that $|H_i(\Pi(x))| \geq \delta$. Then, we have
$$|F_i(x)| \geq (1-\varepsilon/C_n) \delta - C \varepsilon / C_n \geq \delta - (C+1) \varepsilon/C_n > \delta - (\min\{2,1/\varepsilon\}-1) \varepsilon \geq \delta -(\delta-\varepsilon) \geq \varepsilon$$
where we used $C_n > (C+1)/(\min\{2,1/\varepsilon\}-1)$. But this would mean that $\|F(x)\|_\infty > \varepsilon$, a contradiction. Thus, it must be that $\|H(\Pi(x))\|_\infty < \delta$.

In order to show that $\Pi(x)$ is a solution to $H$, it remains to prove that there exists $i \in [n]$ such that $|[\Pi(x)]_i| \geq 1/(n+1)$. Since $x \in \partial([-1,1]^{n+1})$, there exists $j \in [n+1]$ such that $|x_j|=1$. As a result, $|[\Pi(x)]_j| = |x_j - S(x)/(n+1)| \geq 1 - \varepsilon/C_n$. If $j < n+1$, then let $i:=j$. Otherwise, if $j=n+1$, then, since $S(\Pi(x)) = 0$, there necessarily exists $i \in [n]$ such that $|[\Pi(x)]_i| \geq 1/n - \varepsilon/(n C_n)$. In both cases we have found $i \in [n]$ such that $|[\Pi(x)]_i| \geq 1/n - \varepsilon/(n C_n)$. Since $C_n \geq (n+1) \varepsilon$, it follows that $|[\Pi(x)]_i| \geq 1/(n+1)$.

Thus, from any $x \in \partial([-1,1]^{n+1})$ with $\|F(x)\|_\infty \leq \varepsilon$, we can obtain a solution to the original $(n-1)D$-\tucker instance. Note that the reduction is polynomial-time and we have only used operations allowed by the gates of the arithmetic circuit. In particular, we have only used division by constants, which can be performed by multiplying by the inverse of that constant.

The reduction is black-box and any query to $F$ can be answered by at most $2n$ queries to the labelling function $\ell$ of the original $(n-1)D$-\tucker instance. Note that this is a constant, since $n$ is constant.

\section{Relations to the Robertson-Webb Query Model}\label{sec:RW}

\noindent The black-box query model that we used in the previous sections is the standard model used in the related literature of query complexity, where the nature of the input functions depends on the specific problems at hand. For example, for $nD$-\bu the function $F$ inputs points on the domain and returns other points, whereas in $nD$-\tucker the function $\ell$ inputs points and outputs their labels.

At the same time, in the literature of the cake-cutting problem, the predominant query model is in fact a more expressive query model, known as the \emph{Robertson-Webb model (RW)} \citep{robertson1998cake,woeginger2007complexity}. The RW model has been defined only for the case of \emph{additive valuations}, and consists of the following two types of queries:
\begin{itemize}
    \item[-] \eval\ queries, where the agent is given an interval $[a,b]$ and she returns her value for that interval, and
    \item[-] \cut queries, where the agent is given a point $x \in [0,1]$ and a real number $\alpha$, and they designate the smallest interval $[x,y]$, for which their value is exactly $\alpha$.
\end{itemize}
In fact, in the literature of \emph{envy-free} cake-cutting, the query complexity in the RW model has been one of the most important open problems \citep{brams1996fair,Procaccia16-survey}, with breakthrough results coming from the literature of computer science fairly recently \citep{aziz2016discrete,aziz2016discreteb}. Since $\e$-\ch\ and $\e$-\emph{fair} cake-cutting \citep{branzei2017query} are conceptually closely related, it would make sense to consider the query complexity of the former problem in the RW model as well.\footnote{The $\e$-\ch\ halving problem has in fact recently been studied under this query model as well, but in a somewhat different direction, and for agents with additive valuations \citep{AlonG21-efficient}. Note that the authors do not refer to their query model as the RW model, but the queries that they use are essentially RW queries.} 

A potential hurdle in this investigation is that the RW model has not been defined for valuation functions beyond the additive case.  
To this end, we propose the following generalisation of the RW model that we call \emph{Generalised Robertson-Webb model (GRW))}, which is appropriate for monotone valuation functions that are not necessarily additive. Intuitively, in the GRW model the agent essentially is given \emph{sets of intervals} $A$ rather than single intervals, and the queries are defined accordingly (see also \cref{fig:cut-eval}).

\begin{mdframed}[backgroundcolor=white!90!gray,
      leftmargin=\dimexpr\leftmargin-20pt\relax,
      innerleftmargin=4pt,
      innertopmargin=0pt,
      skipabove=5pt,skipbelow=5pt]
\begin{definition}[\emph{Generalised Robertson-Webb (GRW) Query Model}]
In the GRW query model, there are two types of queries:
\begin{itemize}
    \item[-] \eval\ queries, where agent $i$ is given any Lebesgue-measurable subset $A$ of $[0,1]$ and she returns her value $v_i(A)$ for that set, and
    \item[-] \cut\ queries, where agent $i$ is given two disjoint Lebesgue-measurable subsets $A_1$ and $A_2$ of $[0,1]$, an interval $I=[a,b]$, disjoint from $A_1$ and $A_2$, and a real number $\gamma \geq 0$, and she designates some $x \in I$ such that $\frac{v_i(A_1 \cup [a,x])}{v_i(A_2 \cup ([x,b])} = \gamma$, if such a point exists. 
\end{itemize}
\end{definition}
\end{mdframed}\vskip 5pt

\begin{figure}
  \begin{center}
      \includegraphics[width=\textwidth]{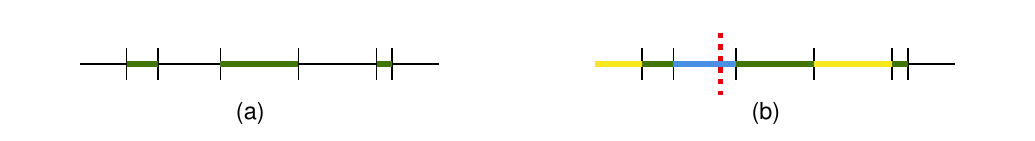}
  \end{center}
    \caption{Visualisation of \cut and \eval queries.  (a) The input $A$ to an \eval query is denoted by the green intervals. (b) The inputs $A_1$ and $A_2$ to the \cut query are denoted by the green and yellow intervals respectively, and the interval $I=[a,b]$ is denoted in blue. The agent places a cut (if possible) at a position $x \in I$ such that her value for $A_1 \cup [a,x]$ and her value for $A_2 \cup [x,b]$ are in a specified proportion.} 
    \label{fig:cut-eval}
\end{figure}

\noindent Let us discuss why this model is the most appropriate generalisation of the RW model. First, the definition of \eval\ queries is in fact the natural extension, as the agent needs to specify her value for sets of intervals; note that in the additive case, it suffices to elicit an agent's value for only single intervals, as her value for unions of intervals is then simply the sum of the elicited values. This is not the case in general for monotone valuations, and therefore we need a more expressive \eval\ query. We also remark that the \eval\ query is exactly the same as a query in the black-box model, as defined in \cref{sec:preliminaries}, and therefore the GRW model is stronger than the black-box query model. \citet{branzei2017query} in fact studied the restriction of the RW model (for the cake-cutting problem and for additive valuations), for which only \eval\ queries are allowed, and they coined this the \emph{$RW^{-}$ query model}. To put our results into context, we offer the following definition of the \emph{$GRW^-$ query model}, which is, as discussed, equivalent to the black-box query model of \cref{sec:preliminaries}.
By the discussion above, all of our query complexity bounds in \cref{sec:general} and \cref{sec:monotone} apply verbatim to the GRW$^-$ query model.
\begin{definition}[\emph{Generalised Robertson-Webb$^-$ (GRW$^-$) query model}]
In the GRW$^-$ query model, only \eval\ queries are allowed; there agent $i$ is given a Lebesgue-measurable subset $A$ of $[0,1]$ and she returns her value $v_i(A)$ for that set.
\end{definition}

While the extension of \eval\ queries from the RW model to the GRW model is relatively straightforward, the generalisation of \cut\ queries is somewhat more intricate. Upon closer inspection of a \cut query in the (standard) RW model for additive valuations, it is clear that one can equivalently define this query as
\begin{quote}
\emph{Given an interval $I=[a,b]$ and a real number $\gamma$, place a cut at $x \in [a,b]$ such that $\frac{v_i([a,x])}{v_i([x,b])} = \gamma$.}
\end{quote}
This is because one can easily find the value of the agent for $[a,b]$ with one \eval query, and then for any value of $\alpha$ used in the standard definition of a cut query, there is an appropriate value of $\gamma$ in the modified definition above, which will result in exactly the same position $x$ of the cut and vice-versa.

The simplicity of the \cut queries in the RW model is enabled by the fact that for additive valuations, the value of any agent $i$ for an interval $I$ does not depend on how the remainder of the interval $[0,1]$ has been cut. This is no longer the case for monotone valuations, as now the agent needs to specify a different value for sets of intervals. We believe that our definition of the \cut query in the GRW model is the appropriate generalisation, which captures the essence of the original cut queries in RW, but also allows for enough expressiveness to make this type of query useful for monotone valuations beyond the additive case. 

Finally, we remark that for general valuations (beyond monotone), any sensible definition of cut queries seems to be too strong, in the sense that it conveys unrealistically too much information (in contrast to the RW and GRW models, where the \cut queries are intuitively ``shortcuts'' for binary search). For example, assume that the agent is asked to place a cut at some point $x$ in an interval $[a,b]$, for which (a) if the cut is placed at $a$ the agent ``sees'' an excess of \lplus and (b) if the cut is placed at $b$, the agent still ``sees'' an excess of \lplus. By the boundary conditions of the interval, there is no guarantee that a cut that ``satisfies'' the agent exists within that interval, and we would need to exhaustively search through the whole interval to find such a cut position, if it exists, meaning that binary search does not help us here. On the other hand, a single \cut query would either find the position or return that there is no such $x$ within the interval. \\

\noindent We are now ready to state our results for the section, which we summarise in the following theorem. Qualitatively, we prove that $\e$-\ch\ with three normalised monotone agents still has an exponential query complexity in the GRW model (with logarithmic savings compared to the black-box model), whereas for two normalised monotone agents, the problem becomes ``easier'' by a logarithmic factor.

\begin{theorem}\label{thm:query-GRW}
In the Generalised Robertson-Webb model:
\begin{itemize}
    \item[-] $\e$-\ch\ with $n \geq 3$ normalised monotone agents requires $\Omega\left(\frac{(L/\e)^{n-2}}{\log(L/\e)}\right)$ queries.
    \item[-] $\e$-\ch\ with $n=2$ monotone agents can be solved with $O(\log(L/\e))$ queries.
\end{itemize}
\end{theorem}

\begin{proof}
The upper bound for $n=2$ can be obtained relatively easily, by observing that in the proof of \cref{thm:CH-monotone-2agent}, Step~\ref{step-w1} and Step~\ref{step-w2} of Algorithm~\ref{alg:two-agents-monotone} were obtained via binary search, using $O(\log (L/\e))$ queries, which resulted in a query complexity of $O(\log^2 (L/\e))$, since these steps were executed $O(\log (L/\e))$ times. In the GRW model, we can simply replace each of those binary searches by a single \cut query (as these only apply to the monotone agents) and obtain a query complexity of $O(\log (L/\e))$. For example, Step~\ref{step-w1} can be simulated by a \cut query where $\gamma = 1$, $A_1= [\frac{x_\ell+x_r}{2},y_\ell]$, $A_2 = [0, \frac{x_\ell+x_r}{2}] \cup [y_r,1]$, and $I = [y_\ell,y_r]$. 

For the lower bound when $n \geq 3$, we will show how to construct an instance of $\e$-\ch\ with $n$ normalised monotone agents, such that the $\Omega((L/\e)^{n-2})$ lower bound for \eval\ queries still holds, but we can additionally answer any \cut\ query by performing at most $O(\log(L/\varepsilon))$ \eval\ queries. We will again use our reduction from normalised monotone $nD$-\bu to $\e$-\ch\ with $n$ normalised monotone agents (\cref{prop:BU2CH}).

Given a \cut\ query $(A_1,A_2,I=[a,b],\gamma)$, we define $\phi: [a,b] \to \mathbb{R}_{\geq 0}$ by $\phi(t) = \frac{v_i(A_1 \cup [a,t])}{v_i(A_2 \cup [t,b])}$. Note that we are looking for $t^* \in I$ such that $\phi(t^*)=\gamma$ and that $\phi$ can be evaluated by using two \eval\ queries. Furthermore, since $v_i$ is monotone, $\phi$ is non-decreasing. We begin by checking that $\frac{v_i(A_1)}{v_i(A_2 \cup [a,b])} \leq \gamma$ and $\frac{v_i(A_1 \cup [a,b])}{v_i(A_2)} \geq \gamma$. If one of these two conditions does not hold, then we can immediately answer that there is no $t \in I$ that satisfies the query. In what follows, we assume that these two conditions hold. In that case, we can query $\phi(t)$ for some $t \in I$ to determine whether the solution $t^*$ lies in $[a,t]$ or in $[t,b]$. We denote by $J \subseteq I$ the current interval for which we know that $t^* \in J$. At the beginning, we have $J:=I$. Using at most $\lceil \log_2(n+1) \rceil + 2$ queries to $\phi$ we can shrink $J$ such that $J \subseteq R_j$ for some $j \in [n+1]$. Recall that $[x(A)]_i = 2 (n+1) \cdot \lambda(A \cap R_i) - 1$ for any $i \in [n+1]$ and $A \in \Lambda([0,1])$. It follows that for any $t \in J$, $[x(A_1 \cup [a,t])]_i$ and $[x(A_2 \cup [t,b])]_i$ are fixed for all $i \in [n+1] \setminus \{j\}$. Furthermore, with an additional $\lceil \log_2 m \rceil + 2$ queries, we can ensure that for all $t \in J$, $[x(A_1 \cup [a,t])]_j \in [k/m,(k+1)/m]$ and $[x(A_2 \cup [t,b])]_j \in [\ell/m,(\ell+1)/m]$ for some $k,\ell \in \mathbb{Z}$. Next, with an additional $2 \lceil \log_2 (n+1) \rceil$ queries, we can shrink $J$, so that for all $t \in J$, $x(A_1 \cup [a,t])$ and $x(A_2 \cup [t,b])$ each lie in some fixed simplex of Kuhn's triangulation of the domain $K_m^{n+1}$ (defined below). In that case, by our construction below, it will hold that $v_i(x(A_1 \cup [a,t]))$ and $v_i(x(A_2 \cup [t,b]))$ can be expressed as an affine function of $t \in J$ and thus we can exactly determine the value of $t^*$. In order for this to hold, we will ensure that our normalised monotone $nD$-\bu function is piecewise linear. Furthermore, we will pick $m=\lceil 2nL/\varepsilon \rceil$, and thus we have used $2(3\lceil \log_2(n+1) \rceil + 2 + \lceil \log_2 (\lceil 2nL/\varepsilon \rceil) \rceil + 2)$ \eval\ queries to answer one \cut\ query. Note that this expression is $O(\log(L/\varepsilon))$, since $n$ is constant.

Consider a normalised monotone $nD$-\bu function $F: [-1,1]^{n+1} \to [-1,1]^n$ with Lipschitz parameter $L \geq 3$ and some $\varepsilon \in (0,1)$. We first discretize the domain to be $K_m^{n+1} := \{-1,-(m-1)/m, \dots, -1/m,0,1/m,2/m, \dots, (m-1)/m,1\}^{n+1}$ where $m=\lceil 2nL/\varepsilon \rceil$. We let $f: K_m^{n+1} \to [-1,1]^n$ be defined by $f(x) = F(x)$. Note that $f$ is antipodally anti-symmetric ($f(-x) = -f(x)$ for all $x \in K_m^{n+1}$), monotone ($f(x) \leq f(y)$ whenever $x \leq y$) and $f(1,1,\dots,1)=(1,1,\dots,1)$. Furthermore, any $x \in \partial(K_m^{n+1})$ with $\|f(x)\|_\infty \leq \varepsilon$ yields a solution to the original instance $F$. We extend $f$ back to a function $\widehat{f}: [-1,1]^{n+1} \to [-1,1]^n$ by using Kuhn's triangulation on the grid $K_m^{n+1}$ and interpolating (see \cref{sec:kuhn} for a description of the triangulation and interpolation). By the arguments presented in \cref{sec:kuhn}, it holds that $\widehat{f}$ is a continuous, monotone, odd function, and $\widehat{f}(1,1,\dots,1) = (1,1,\dots,1)$. Furthermore, if $x_i$ is fixed for all $i \in [n+1] \setminus \{j\}$ and $x_j$ is constrained such that $x$ lies in some fixed simplex $\sigma$ of Kuhn's triangulation, then $\widehat{f}(x)$ can be expressed as a linear affine function of $x_j$.

Let us now determine the Lipschitz parameter of $\widehat{f}$. Consider any simplex $\sigma = \{y^0$, $y^1$, $\dots$, $y^{n+1}\}$ of the Kuhn triangulation of $K_m^{n+1}$. Consider any $x \in [0,1]^{n+1}$ that lies in $\sigma$ and any $j \in [n+1]$ and $t \in [-1/m,1/m]$ such that $x+t\cdot e_j$ also lies in $\sigma$. Then, the interpolation (as defined in \cref{sec:kuhn}) yields
\begin{equation*}
\begin{split}
\|\widehat{f}(x) - \widehat{f}(x+t\cdot e_j)\|_\infty = \|t\cdot m \cdot f(y^j) - t \cdot m \cdot f(y^{j-1})\|_\infty &\leq |t| \cdot m \cdot \|F(y^j) - F(y^{j-1})\|_\infty\\
&\leq |t| \cdot m \cdot L \cdot \|y^j - y^{j-1}\|_\infty\\
&\leq L \cdot |t|.
\end{split}
\end{equation*}
It is easy to check that this implies that $\widehat{f}$ is $(n+1)L$-Lipschitz-continuous on the simplex $\sigma$. By a simple argument, it follows that $\widehat{f}$ is $(n+1)L$-Lipschitz-continuous on $[-1,1]^{n+1}$ (see e.g., the proof of \cref{lem:TUtoBU-general}).

Now consider any simplex $\sigma = \{y^0,y^1,\dots,y^{n+1}\}$ such that there exists $i \in [n+1] \cup \{0\}$ with $\|f(y^i)\|_\infty > \varepsilon$. Since $\widehat{f}(y^i) = f(y^i)$, and $\widehat{f}$ is $(n+1)L$-Lipschitz-continuous, it follows that for any $x \in [-1,1]^{n+1}$ that lies in $\sigma$ we have
$$\|\widehat{f}(x)\|_\infty \geq \|\widehat{f}(y^i)\|_\infty - nL \|x-y^0\|_\infty > \varepsilon - nL/m \geq \varepsilon/2$$
where we used $m \geq 2nL/\varepsilon$. Thus, for any $x \in \partial([-1,1]^{n+1})$ with $\|\widehat{f}(x)\|_\infty \leq \varepsilon/2$, it must hold that both $\|f(y^0)\|_\infty \leq \varepsilon$ and $\|f(y^{n+1})\|_\infty \leq \varepsilon$, where $\sigma = \{y^0,y^1,\dots,y^{n+1}\}$ is the Kuhn simplex containing $x$. However, since $x \in \partial([-1,1]^{n+1})$, it follows that $y^0$ or $y^{n+1}$ lies on the boundary of $K_m^{n+1}$. This means that we have obtained a solution to the original $nD$-\bu function $F$.

Since the parameters for $\widehat{f}$ are $L'=nL$ and $\varepsilon'=\varepsilon/2$, and $n$ is constant, the query lower bound for $F$ carries over to $\widehat{f}$.
\end{proof}

\section{Conclusion and Future Directions}

In this paper, we managed to completely settle the computational complexity of the $\e$-\ch\ problem for a constant number of agents with either general or monotone valuation functions. We also studied the query complexity of the problem and we provided exponential lower bounds corresponding to our hardness results, and polynomial upper bounds corresponding to our polynomial-time algorithms. We also defined an appropriate generalisation of the Robertson-Webb query model for monotone valuations and we showed that our bounds are qualitatively robust to the added expressiveness of this model. The main open problem associated with our work is the following.

\begin{quote}
What is the computational complexity and the query complexity of $\e$-\ch\ with a \emph{constant} number of agents and \emph{additive valuations}?
\end{quote}
Our approach in this paper prescribes a formula for answering this question: One can construct a black-box reduction to this version of $\e$-\ch\ from a computationally-hard problem like $nD$-\tucker, for which we also know query complexity lower bounds, and obtain answers to both questions at the same time. Alternatively, one might be able to construct polynomial-time algorithms for solving this problem; concretely, attempting to do that for three agents with additive valuations would be the natural starting step, as this is the first case for which the problem becomes computationally hard for agents with monotone valuations. It is unclear whether one should expect the problem to remain hard for additive valuations.

Another line of work would be to study the query complexity of the related fair cake-cutting problem using the GRW model that we propose. In fact, while the fundamental existence results for the problem (e.g., see \citep{su1999rental}) actually apply to quite general valuation functions, most of the work in computer science has restricted its attention to the case of additive valuations only, with a few notable exceptions \citep{caragiannis2011towards,deng2012algorithmic}.
We believe that our work can therefore spark some interest in the study of \emph{the query complexity of fair cake-cutting with valuations beyond the additive case}.

\subsubsection*{Acknowledgements}
Alexandros Hollender was supported by an EPSRC doctoral studentship (Reference 1892947).

\bibliographystyle{plainnat}
\bibliography{references-journal}

\begin{thebibliography}{67}
\providecommand{\natexlab}[1]{#1}
\providecommand{\url}[1]{\texttt{#1}}
\expandafter\ifx\csname urlstyle\endcsname\relax
  \providecommand{\doi}[1]{doi: #1}\else
  \providecommand{\doi}{doi: \begingroup \urlstyle{rm}\Url}\fi

\bibitem[Aisenberg et~al.(2020)Aisenberg, Bonet, and Buss]{ABB15-2DTucker}
James Aisenberg, Maria~Luisa Bonet, and Sam Buss.
\newblock 2-{D} {T}ucker is {PPA} complete.
\newblock \emph{Journal of Computer and System Sciences}, 108:\penalty0
  92--103, 2020.
\newblock \doi{10.1016/j.jcss.2019.09.002}.

\bibitem[Alijani et~al.(2017)Alijani, Farhadi, Ghodsi, Seddighin, and
  Tajik]{alijani2017envy}
Reza Alijani, Majid Farhadi, Mohammad Ghodsi, Masoud Seddighin, and Ahmad
  Tajik.
\newblock Envy-free mechanisms with minimum number of cuts.
\newblock In \emph{Proceedings of the 31st AAAI Conference on Artificial
  Intelligence (AAAI)}, pages 312--318, 2017.
\newblock \doi{10.1609/aaai.v31i1.10584}.

\bibitem[Alon(1987)]{Alon87-Necklace}
Noga Alon.
\newblock Splitting necklaces.
\newblock \emph{Advances in Mathematics}, 63\penalty0 (3):\penalty0 247--253,
  1987.
\newblock \doi{10.1016/0001-8708(87)90055-7}.

\bibitem[Alon and Graur(2021)]{AlonG21-efficient}
Noga Alon and Andrei Graur.
\newblock {Efficient Splitting of Necklaces}.
\newblock In \emph{Proceedings of the 48th International Colloquium on
  Automata, Languages, and Programming (ICALP)}, pages 14:1--14:17, 2021.
\newblock \doi{10.4230/LIPIcs.ICALP.2021.14}.

\bibitem[Alon and West(1986)]{Alon1986}
Noga Alon and Douglas~B. West.
\newblock {The Borsuk-Ulam Theorem and Bisection of Necklaces}.
\newblock \emph{Proceedings of the American Mathematical Society}, 98\penalty0
  (4):\penalty0 623--628, 1986.
\newblock \doi{10.2307/2045739}.

\bibitem[Amanatidis et~al.(2018)Amanatidis, Christodoulou, Fearnley, Markakis,
  Psomas, and Vakaliou]{amanatidis2018improved}
Georgios Amanatidis, George Christodoulou, John Fearnley, Evangelos Markakis,
  Christos-Alexandros Psomas, and Eftychia Vakaliou.
\newblock An improved envy-free cake cutting protocol for four agents.
\newblock In \emph{Proceedings of the 11th International Symposium on
  Algorithmic Game Theory (SAGT)}, pages 87--99, 2018.
\newblock \doi{10.1007/978-3-319-99660-8\_9}.

\bibitem[Austin(1982)]{austin1982sharing}
A.~K. Austin.
\newblock Sharing a cake.
\newblock \emph{The Mathematical Gazette}, 66\penalty0 (437):\penalty0
  212--215, 1982.
\newblock \doi{10.2307/3616548}.

\bibitem[Aziz and Mackenzie(2016{\natexlab{a}})]{aziz2016discrete}
Haris Aziz and Simon Mackenzie.
\newblock A discrete and bounded envy-free cake cutting protocol for any number
  of agents.
\newblock In \emph{Proceedings of the 57th Annual IEEE Symposium on Foundations
  of Computer Science (FOCS)}, pages 416--427, 2016{\natexlab{a}}.
\newblock \doi{10.1109/FOCS.2016.52}.

\bibitem[Aziz and Mackenzie(2016{\natexlab{b}})]{aziz2016discreteb}
Haris Aziz and Simon Mackenzie.
\newblock A discrete and bounded envy-free cake cutting protocol for four
  agents.
\newblock In \emph{Proceedings of the 48th Annual ACM Symposium on Theory of
  Computing (STOC)}, pages 454--464, 2016{\natexlab{b}}.
\newblock \doi{10.1145/2897518.2897522}.

\bibitem[Balkanski et~al.(2014)Balkanski, Br{\^a}nzei, Kurokawa, and
  Procaccia]{balkanski2014simultaneous}
Eric Balkanski, Simina Br{\^a}nzei, David Kurokawa, and Ariel~D. Procaccia.
\newblock Simultaneous cake cutting.
\newblock In \emph{Proceedings of the 28th AAAI Conference on Artificial
  Intelligence (AAAI)}, pages 566--572, 2014.
\newblock \doi{10.1609/aaai.v28i1.8802}.

\bibitem[Barman and Rathi(2020)]{BarmanR20cake}
Siddharth Barman and Nidhi Rathi.
\newblock Fair cake division under monotone likelihood ratios.
\newblock In \emph{Proceedings of the 21st ACM Conference on Economics and
  Computation (EC)}, pages 401--437, 2020.
\newblock \doi{10.1145/3391403.3399512}.

\bibitem[Batziou et~al.(2021)Batziou, Hansen, and
  H{\o}gh]{BatziouHH21-consensus-BBU}
Eleni Batziou, Kristoffer~Arnsfelt Hansen, and Kasper H{\o}gh.
\newblock Strong approximate {C}onsensus {H}alving and the {B}orsuk-{U}lam
  theorem.
\newblock In \emph{Proceedings of the 48th International Colloquium on
  Automata, Languages, and Programming (ICALP)}, pages 24:1--24:20, 2021.
\newblock \doi{10.4230/LIPIcs.ICALP.2021.24}.

\bibitem[Beame et~al.(1998)Beame, Cook, Edmonds, Impagliazzo, and
  Pitassi]{Beame1998}
Paul Beame, Stephen Cook, Jeff Edmonds, Russell Impagliazzo, and Toniann
  Pitassi.
\newblock {The Relative Complexity of NP Search Problems}.
\newblock \emph{Journal of Computer and System Sciences}, 57\penalty0
  (1):\penalty0 3--19, 1998.
\newblock \doi{10.1145/225058.225147}.

\bibitem[Bei and Suksompong(2021)]{bei2021dividing}
Xiaohui Bei and Warut Suksompong.
\newblock Dividing a graphical cake.
\newblock In \emph{Proceedings of the 35th AAAI Conference on Artificial
  Intelligence (AAAI)}, pages 5159--5166, 2021.
\newblock URL \url{https://ojs.aaai.org/index.php/AAAI/article/view/16652}.

\bibitem[Bei et~al.(2012)Bei, Chen, Hua, Tao, and Yang]{bei2012optimal}
Xiaohui Bei, Ning Chen, Xia Hua, Biaoshuai Tao, and Endong Yang.
\newblock Optimal proportional cake cutting with connected pieces.
\newblock In \emph{Proceedings of the 26th AAAI Conference on Artificial
  Intelligence (AAAI)}, pages 1263--1269, 2012.
\newblock \doi{10.1609/aaai.v26i1.8243}.

\bibitem[Bei et~al.(2017)Bei, Chen, Huzhang, Tao, and Wu]{bei2017cake}
Xiaohui Bei, Ning Chen, Guangda Huzhang, Biaoshuai Tao, and Jiajun Wu.
\newblock Cake cutting: Envy and truth.
\newblock In \emph{Proceedings of the 26th International Joint Conference on
  Artificial Intelligence (IJCAI)}, pages 3625--3631, 2017.
\newblock \doi{10.24963/ijcai.2017/507}.

\bibitem[Bei et~al.(2022)Bei, Igarashi, Lu, and Suksompong]{bei2021price}
Xiaohui Bei, Ayumi Igarashi, Xinhang Lu, and Warut Suksompong.
\newblock The price of connectivity in fair division.
\newblock \emph{SIAM Journal on Discrete Mathematics}, 36\penalty0
  (2):\penalty0 1156--1186, 2022.
\newblock \doi{10.1137/20M1388310}.

\bibitem[Borsuk(1933)]{Borsuk1933}
Karol Borsuk.
\newblock {Drei S{\"{a}}tze {\"{u}}ber die n-dimensionale euklidische
  Sph{\"{a}}re}.
\newblock \emph{Fundamenta Mathematicae}, 20:\penalty0 177--190, 1933.
\newblock \doi{10.4064/fm-20-1-177-190}.

\bibitem[Bouveret et~al.(2016)Bouveret, Chevaleyre, and
  Maudet]{BouveretCM16-indivisible}
Sylvain Bouveret, Yann Chevaleyre, and Nicolas Maudet.
\newblock Fair allocation of indivisible goods.
\newblock In \emph{Handbook of Computational Social Choice}, pages 284--310.
  Cambridge University Press, 2016.
\newblock \doi{10.1017/CBO9781107446984.013}.

\bibitem[Brams and Taylor(1996)]{brams1996fair}
Steven~J. Brams and Alan~D. Taylor.
\newblock \emph{Fair Division: From cake-cutting to dispute resolution}.
\newblock Cambridge University Press, 1996.
\newblock \doi{10.1017/CBO9780511598975}.

\bibitem[Br{\^a}nzei and Nisan(2017)]{branzei2017query}
Simina Br{\^a}nzei and Noam Nisan.
\newblock The query complexity of cake cutting.
\newblock \emph{arXiv preprint}, 2017.
\newblock URL \url{https://arxiv.org/abs/arXiv:1705.02946}.

\bibitem[Br{\^a}nzei and Nisan(2019)]{branzei2019communication}
Simina Br{\^a}nzei and Noam Nisan.
\newblock Communication complexity of cake cutting.
\newblock In \emph{Proceedings of the 20th ACM Conference on Economics and
  Computation (EC)}, pages 525--525, 2019.
\newblock \doi{10.1145/3328526.3329644}.

\bibitem[Br{\^a}nzei et~al.(2016)Br{\^a}nzei, Caragiannis, Kurokawa, and
  Procaccia]{branzei2016algorithmic}
Simina Br{\^a}nzei, Ioannis Caragiannis, David Kurokawa, and Ariel~D.
  Procaccia.
\newblock An algorithmic framework for strategic fair division.
\newblock In \emph{Proceedings of the 30th AAAI Conference on Artificial
  Intelligence (AAAI)}, pages 411--417, 2016.
\newblock URL
  \url{https://www.aaai.org/ocs/index.php/AAAI/AAAI16/paper/view/12294}.

\bibitem[Caragiannis et~al.(2011)Caragiannis, Lai, and
  Procaccia]{caragiannis2011towards}
Ioannis Caragiannis, John~K. Lai, and Ariel~D. Procaccia.
\newblock Towards more expressive cake cutting.
\newblock In \emph{Proceedings of the 22nd International Joint Conference on
  Artificial Intelligence (IJCAI)}, pages 127--132, 2011.
\newblock \doi{10.5591/978-1-57735-516-8/IJCAI11-033}.

\bibitem[Chen and Deng(2008)]{chen2008matching}
Xi~Chen and Xiaotie Deng.
\newblock Matching algorithmic bounds for finding a {B}rouwer fixed point.
\newblock \emph{Journal of the ACM}, 55\penalty0 (3):\penalty0 13:1--13:26,
  2008.
\newblock \doi{10.1145/1379759.1379761}.

\bibitem[Chen et~al.(2009)Chen, Deng, and Teng]{chen2009settling}
Xi~Chen, Xiaotie Deng, and Shang-Hua Teng.
\newblock Settling the complexity of computing two-player {N}ash equilibria.
\newblock \emph{Journal of the ACM}, 56\penalty0 (3):\penalty0 14:1--14:57,
  2009.
\newblock \doi{10.1145/1516512.1516516}.

\bibitem[Cohler et~al.(2011)Cohler, Lai, Parkes, and
  Procaccia]{cohler2011optimal}
Yuga~J. Cohler, John~K. Lai, David~C. Parkes, and Ariel~D. Procaccia.
\newblock Optimal envy-free cake cutting.
\newblock In \emph{Proceedings of the 25th AAAI Conference on Artificial
  Intelligence (AAAI)}, pages 626--631, 2011.
\newblock URL
  \url{https://www.aaai.org/ocs/index.php/AAAI/AAAI11/paper/viewPaper/3638}.

\bibitem[Daskalakis and Papadimitriou(2011)]{DaskalakisP11-CLS}
Constantinos Daskalakis and Christos Papadimitriou.
\newblock Continuous local search.
\newblock In \emph{Proceedings of the 22nd Annual ACM-SIAM Symposium on
  Discrete Algorithms (SODA)}, pages 790--804. SIAM, 2011.
\newblock \doi{10.1137/1.9781611973082.62}.

\bibitem[Daskalakis et~al.(2009)Daskalakis, Goldberg, and
  Papadimitriou]{Daskalakis2009}
Constantinos Daskalakis, Paul~W. Goldberg, and Christos~H. Papadimitriou.
\newblock The complexity of computing a {N}ash equilibrium.
\newblock \emph{SIAM Journal on Computing}, 39\penalty0 (1):\penalty0 195--259,
  2009.
\newblock \doi{10.1137/070699652}.

\bibitem[Deligkas et~al.(2021)Deligkas, Fearnley, Melissourgos, and
  Spirakis]{deligkas2019computing}
Argyrios Deligkas, John Fearnley, Themistoklis Melissourgos, and Paul~G.
  Spirakis.
\newblock Computing exact solutions of consensus halving and the
  {B}orsuk-{U}lam theorem.
\newblock \emph{Journal of Computer and System Sciences}, 117:\penalty0 75--98,
  2021.
\newblock \doi{10.1016/j.jcss.2020.10.006}.

\bibitem[Deligkas et~al.(2022{\natexlab{a}})Deligkas, Fearnley, Hollender, and
  Melissourgos]{DeligkasFHM2022-CH-constant-eps}
Argyrios Deligkas, John Fearnley, Alexandros Hollender, and Themistoklis
  Melissourgos.
\newblock Constant inapproximability for {PPA}.
\newblock In \emph{Proceedings of the 54th ACM Symposium on Theory of Computing
  (STOC)}, pages 1010--1023, 2022{\natexlab{a}}.
\newblock \doi{10.1145/3519935.3520079}.

\bibitem[Deligkas et~al.(2022{\natexlab{b}})Deligkas, Fearnley, and
  Melissourgos]{DFM-pizza}
Argyrios Deligkas, John Fearnley, and Themistoklis Melissourgos.
\newblock Pizza sharing is {PPA}-hard.
\newblock In \emph{Proceedings of the 36th AAAI Conference on Artificial
  Intelligence (AAAI)}, pages 4957--4965, 2022{\natexlab{b}}.
\newblock \doi{10.1609/aaai.v36i5.20426}.

\bibitem[Deng et~al.(2011)Deng, Qi, Saberi, and Zhang]{deng2011discrete}
Xiaotie Deng, Qi~Qi, Amin Saberi, and Jie Zhang.
\newblock Discrete fixed points: Models, complexities, and applications.
\newblock \emph{Mathematics of Operations Research}, 36\penalty0 (4):\penalty0
  636--652, 2011.
\newblock \doi{10.1287/moor.1110.0511}.

\bibitem[Deng et~al.(2012)Deng, Qi, and Saberi]{deng2012algorithmic}
Xiaotie Deng, Qi~Qi, and Amin Saberi.
\newblock Algorithmic solutions for envy-free cake cutting.
\newblock \emph{Operations Research}, 60\penalty0 (6):\penalty0 1461--1476,
  2012.
\newblock \doi{10.1287/opre.1120.1116}.

\bibitem[Elkind et~al.(2021)Elkind, Segal-Halevi, and
  Suksompong]{elkind2021mind}
Edith Elkind, Erel Segal-Halevi, and Warut Suksompong.
\newblock Mind the gap: Cake cutting with separation.
\newblock In \emph{Proceedings of the 35th AAAI Conference on Artificial
  Intelligence (AAAI)}, pages 5330--5338, 2021.
\newblock URL \url{https://ojs.aaai.org/index.php/AAAI/article/view/16672}.

\bibitem[Etessami and Yannakakis(2010)]{EY10-Nash-FIXP}
Kousha Etessami and Mihalis Yannakakis.
\newblock On the complexity of {N}ash equilibria and other fixed points.
\newblock \emph{SIAM Journal on Computing}, 39\penalty0 (6):\penalty0
  2531--2597, 2010.
\newblock \doi{10.1137/080720826}.

\bibitem[Fearnley et~al.(2021)Fearnley, Goldberg, Hollender, and
  Savani]{FearnleyGHS-gradient}
John Fearnley, Paul~W. Goldberg, Alexandros Hollender, and Rahul Savani.
\newblock The complexity of gradient descent: {CLS} = {PPAD} $\cap$ {PLS}.
\newblock In \emph{Proceedings of the 53rd ACM Symposium on Theory of Computing
  (STOC)}, pages 46--59, 2021.
\newblock \doi{10.1145/3406325.3451052}.

\bibitem[Filos-Ratsikas and Goldberg(2018)]{FRG18-Consensus}
Aris Filos-Ratsikas and Paul~W. Goldberg.
\newblock {Consensus Halving is {PPA}-complete}.
\newblock In \emph{Proceedings of the 50th Annual ACM Symposium on Theory of
  Computing (STOC)}, pages 51--64, 2018.
\newblock \doi{10.1145/3188745.3188880}.

\bibitem[Filos-Ratsikas and Goldberg(2019)]{FRG18-Necklace}
Aris Filos-Ratsikas and Paul~W. Goldberg.
\newblock {The Complexity of Splitting Necklaces and Bisecting Ham Sandwiches}.
\newblock In \emph{Proceedings of the 51st Annual ACM Symposium on Theory of
  Computing (STOC)}, pages 638--649, 2019.
\newblock \doi{10.1145/3313276.3316334}.

\bibitem[Filos-Ratsikas et~al.(2018)Filos-Ratsikas, Frederiksen, Goldberg, and
  Zhang]{filos2018hardness}
Aris Filos-Ratsikas, S{\o}ren Kristoffer~Still Frederiksen, Paul~W. Goldberg,
  and Jie Zhang.
\newblock {Hardness Results for Consensus-Halving}.
\newblock In \emph{Proceedings of the 43rd International Symposium on
  Mathematical Foundations of Computer Science (MFCS)}, pages 24:1--24:16,
  2018.
\newblock \doi{10.4230/LIPIcs.MFCS.2018.24}.

\bibitem[Filos{-}Ratsikas et~al.(2020)Filos{-}Ratsikas, Hollender, Sotiraki,
  and Zampetakis]{FRHSZ2020consensus-easier}
Aris Filos{-}Ratsikas, Alexandros Hollender, Katerina Sotiraki, and Manolis
  Zampetakis.
\newblock {Consensus-Halving: Does it Ever Get Easier?}
\newblock In \emph{Proceedings of the 21st ACM Conference on Economics and
  Computation (EC)}, pages 381--399, 2020.
\newblock \doi{10.1145/3391403.3399527}.

\bibitem[Filos-Ratsikas et~al.(2021)Filos-Ratsikas, Hollender, Sotiraki, and
  Zampetakis]{filos2020topological}
Aris Filos-Ratsikas, Alexandros Hollender, Katerina Sotiraki, and Manolis
  Zampetakis.
\newblock A topological characterization of modulo-p arguments and implications
  for necklace splitting.
\newblock In \emph{Proceedings of the 32nd Annual ACM-SIAM Symposium on
  Discrete Algorithms (SODA)}, pages 2615--2634, 2021.
\newblock \doi{10.1137/1.9781611976465.155}.

\bibitem[Gamow and Stern(1958)]{gamow1958puzzle}
George Gamow and Marvin Stern.
\newblock \emph{Puzzle-math}.
\newblock Viking, 1958.

\bibitem[Goldberg and West(1985)]{Goldberg1985}
Charles~H. Goldberg and Douglas~B. West.
\newblock {Bisection of Circle Colorings}.
\newblock \emph{SIAM Journal on Algebraic Discrete Methods}, 6\penalty0
  (1):\penalty0 93--106, 1985.
\newblock \doi{10.1137/0606010}.

\bibitem[Goldberg et~al.(2020)Goldberg, Hollender, and
  Suksompong]{goldberg2020contiguous}
Paul Goldberg, Alexandros Hollender, and Warut Suksompong.
\newblock Contiguous cake cutting: Hardness results and approximation
  algorithms.
\newblock \emph{Journal of Artificial Intelligence Research}, 69:\penalty0
  109--141, 2020.
\newblock \doi{10.1613/jair.1.12222}.

\bibitem[Goldberg et~al.(2022)Goldberg, Hollender, Igarashi, Manurangsi, and
  Suksompong]{GoldbergHIMS20-consensus-items}
Paul~W. Goldberg, Alexandros Hollender, Ayumi Igarashi, Pasin Manurangsi, and
  Warut Suksompong.
\newblock Consensus halving for sets of items.
\newblock \emph{Mathematics of Operations Research}, 2022.
\newblock \doi{10.1287/moor.2021.1249}.

\bibitem[Grigni(2001)]{Grigni2001}
Michelangelo Grigni.
\newblock {A Sperner lemma complete for PPA}.
\newblock \emph{Information Processing Letters}, 77\penalty0 (5--6):\penalty0
  255--259, 2001.
\newblock \doi{10.1016/S0020-0190(00)00152-6}.

\bibitem[Hobby and Rice(1965)]{hobby1965moment}
Charles~R. Hobby and John~R. Rice.
\newblock {A moment problem in L1 approximation}.
\newblock \emph{Proceedings of the American Mathematical Society}, 16\penalty0
  (4):\penalty0 665--670, 1965.
\newblock \doi{10.2307/2033900}.

\bibitem[Hosseini et~al.(2020)Hosseini, Igarashi, and Searns]{hosseini2020fair}
Hadi Hosseini, Ayumi Igarashi, and Andrew Searns.
\newblock Fair division of time: Multi-layered cake cutting.
\newblock In \emph{Proceedings of the 29th International Joint Conference on
  Artificial Intelligence (IJCAI)}, pages 182--188, 2020.
\newblock \doi{10.24963/ijcai.2020/26}.

\bibitem[Karasev et~al.(2016)Karasev, Rold{\'a}n-Pensado, and
  Sober{\'o}n]{karasev2016measure}
Roman~N. Karasev, Edgardo Rold{\'a}n-Pensado, and Pablo Sober{\'o}n.
\newblock Measure partitions using hyperplanes with fixed directions.
\newblock \emph{Israel Journal of Mathematics}, 212\penalty0 (2):\penalty0
  705--728, 2016.
\newblock \doi{10.1007/s11856-016-1303-z}.

\bibitem[Komargodski et~al.(2019)Komargodski, Naor, and
  Yogev]{komargodski2019white}
Ilan Komargodski, Moni Naor, and Eylon Yogev.
\newblock White-box vs. black-box complexity of search problems: Ramsey and
  graph property testing.
\newblock \emph{Journal of the ACM}, 66\penalty0 (5):\penalty0 1--28, 2019.
\newblock \doi{10.1109/FOCS.2017.63}.

\bibitem[Kurokawa et~al.(2013)Kurokawa, Lai, and Procaccia]{kurokawa2013cut}
David Kurokawa, John Lai, and Ariel Procaccia.
\newblock How to cut a cake before the party ends.
\newblock In \emph{Proceedings of the 27th AAAI Conference on Artificial
  Intelligence (AAAI)}, pages 555--561, 2013.
\newblock \doi{10.1609/aaai.v27i1.8629}.

\bibitem[Matou\v{s}ek(2008)]{Mat03BorsukUlam}
Ji\v{r}{\'i} Matou\v{s}ek.
\newblock \emph{{Using the {B}orsuk-{U}lam theorem: lectures on topological
  methods in combinatorics and geometry}}.
\newblock Springer Science \& Business Media, 2008.
\newblock \doi{10.1007/978-3-540-76649-0}.

\bibitem[Megiddo and Papadimitriou(1991)]{Megiddo1991}
Nimrod Megiddo and Christos~H. Papadimitriou.
\newblock {On total functions, existence theorems and computational
  complexity}.
\newblock \emph{Theoretical Computer Science}, 81\penalty0 (2):\penalty0
  317--324, 1991.
\newblock \doi{10.1016/0304-3975(91)90200-L}.

\bibitem[Meunier(2014)]{Meunier2014simplotopal}
Fr{\'e}d{\'e}ric Meunier.
\newblock Simplotopal maps and necklace splitting.
\newblock \emph{Discrete Mathematics}, 323\penalty0 (28):\penalty0 14--26,
  2014.
\newblock \doi{10.1016/j.disc.2014.01.008}.

\bibitem[Neyman(1946)]{neyman1946theoreme}
Jerzy Neyman.
\newblock Un th{\'e}or{\`e}me d'existence.
\newblock \emph{CR Acad. Sci. Paris}, 222:\penalty0 843--845, 1946.

\bibitem[Papadimitriou(1994)]{Papadimitriou94-TFNP-subclasses}
Christos~H. Papadimitriou.
\newblock On the complexity of the parity argument and other inefficient proofs
  of existence.
\newblock \emph{Journal of Computer and System Sciences}, 48\penalty0
  (3):\penalty0 498--532, 1994.
\newblock \doi{10.1016/S0022-0000(05)80063-7}.

\bibitem[Procaccia(2013)]{procaccia2013cake}
Ariel~D. Procaccia.
\newblock Cake cutting: Not just child's play.
\newblock \emph{Communications of the ACM}, 56\penalty0 (7):\penalty0 78--87,
  2013.
\newblock \doi{10.1145/2483852.2483870}.

\bibitem[Procaccia(2016)]{Procaccia16-survey}
Ariel~D. Procaccia.
\newblock Cake cutting algorithms.
\newblock In \emph{Handbook of Computational Social Choice}, pages 311--330.
  Cambridge University Press, 2016.
\newblock \doi{10.1017/CBO9781107446984.014}.

\bibitem[Robertson and Webb(1995)]{robertson1995approximating}
Jack~M. Robertson and William~A. Webb.
\newblock Approximating fair division with a limited number of cuts.
\newblock \emph{Journal of Combinatorial Theory, Series A}, 72\penalty0
  (2):\penalty0 340--344, 1995.
\newblock \doi{10.1016/0097-3165(95)90073-X}.

\bibitem[Robertson and Webb(1998)]{robertson1998cake}
Jack~M. Robertson and William~A. Webb.
\newblock \emph{Cake-cutting algorithms: Be fair if you can}.
\newblock CRC Press, 1998.
\newblock ISBN 9781568810768.

\bibitem[Segal-Halevi(2022)]{segal2018redividing}
Erel Segal-Halevi.
\newblock Redividing the cake.
\newblock \emph{Autonomous Agents and Multi-Agent Systems}, 36\penalty0 (14),
  2022.
\newblock \doi{10.1007/s10458-022-09545-x}.

\bibitem[Simmons and Su(2003)]{SS03-Consensus}
Forest~W. Simmons and Francis~E. Su.
\newblock Consensus-halving via theorems of {B}orsuk-{U}lam and {T}ucker.
\newblock \emph{Mathematical social sciences}, 45\penalty0 (1):\penalty0
  15--25, 2003.
\newblock \doi{10.1016/S0165-4896(02)00087-2}.

\bibitem[Steinhaus(1949)]{steinhaus1949division}
Hugo Steinhaus.
\newblock Sur la division pragmatique.
\newblock \emph{Econometrica}, 17:\penalty0 315--319, 1949.
\newblock \doi{10.2307/1907319}.

\bibitem[Su(1999)]{su1999rental}
Francis~E. Su.
\newblock Rental harmony: Sperner's lemma in fair division.
\newblock \emph{The American Mathematical Monthly}, 106\penalty0 (10):\penalty0
  930--942, 1999.
\newblock \doi{10.1080/00029890.1999.12005142}.

\bibitem[Tucker(1945)]{tucker1945some}
Albert~W. Tucker.
\newblock Some topological properties of disk and sphere.
\newblock In \emph{Proceedings of the First Canadian Math. Congress, Montreal},
  pages 286--309. University of Toronto Press, 1945.

\bibitem[Woeginger and Sgall(2007)]{woeginger2007complexity}
Gerhard~J. Woeginger and Ji{\v{r}}{\'\i} Sgall.
\newblock On the complexity of cake cutting.
\newblock \emph{Discrete Optimization}, 4\penalty0 (2):\penalty0 213--220,
  2007.
\newblock \doi{10.1016/j.disopt.2006.07.003}.

\end{thebibliography}

\clearpage

\appendix

\section{\texorpdfstring{$\boldsymbol{\e}$}{\e}-\ch\ for piecewise constant valuations} \label{app:piecewiseconstant}

As we mentioned in the introduction, the valuation functions studied in previous work \citep{FRG18-Consensus,FRG18-Necklace} are piecewise constant valuations (i.e., step functions) which are given explicitly as part of the input, with each agent specifying the end-points and the height of each step. For a constant number of agents $n$, this case is solvable in polynomial-time, as illustrated in \cref{fig:taxonomy}.

\begin{lemma}\label{lem:piecewiseconstant}
$\e$-\ch\ with a constant number $n$ of agents and piecewise constant valuations (given explicitly as part of the input) is solvable in polynomial time.
\end{lemma}

\begin{proof}
The proof follows closely that of Lemma 15 in \citep{FRHSZ2020consensus-easier}; here we provide the main proof idea, and we refer the reader to that paper for a fully detailed proof.

First, we partition the interval $[0,1]$ into \emph{regions}, via a set of points $t_1, \ldots, t_m$ such that the density of the valuation function of each agent is constant within each region. Let $T=[t_j,t_{j+1}]$ denote an arbitrary region, and let $v_i(T)$ denote the constant value of the density of agent $i$ in region $T$ (i.e., the height of the step). Since the valuations are provided explicitly (or equivalently, are polynomial in the other input parameters), this process will result in a polynomial number of such regions. 

We will be interested in the regions that will contain cuts, and then we will find the appropriate positions of those cuts using linear programming. First, it is not hard to see that a solution in which some region contains more than $1$ cut can be transformed into a solution in which every region contains at most $1$ cut, by appropriately ``merging'' and ``shifting'' sets of cuts within the region. Therefore, we can assume that each region either contains a cut or it doesn't. For any $k=1,\ldots,n$, we can consider all the possible ways of distributing the $k$ cuts to the regions, such that no region receives more than one cut; since $n$ is a constant, this can be done in polynomial time. 

Let $x_T$ be the position of the cut in region $T=[t_i,t_{i+1}]$. This means that for agent $i$, the ``left sub-region'' $[t_i, x_T]$ receives one label (say \lplus), whereas the ``right sub-region'' $[x_T, t_{i+1}]$ receives the other label (say \lminus), resulting in values $v_i(T)\cdot[t_i, x_T]$ and $v_i(T)\cdot[x_T, t_{i+1}]$ for the agent respectively. If there is not cut in region $T$, then the whole interval receives the same label. Now, we can consider the set of cuts at positions $x_T$ in each of the regions that contain cuts, and the corresponding partitioning of $[0,1]$ into intervals, labelled \lplus and \lminus; without loss of generality we can assume that this labelling is alternating. From there, we can devise a linear program for each value of $k$, where we aim to minimise $z$, subject to $|v_i(\calI^{+}) - v_i(\calI^{-})| \leq z$ (which can be transformed into a set of linear constraints), plus additional linear constraints for the positions of the cuts. Since a solution always exists, and since we consider all values of $k \in [1,n]$, one of the linear programs will terminate with $z=0$, and therefore we can find an \emph{exact} solution to the problem. 
\end{proof}

\section{Formal definition of the input model}\label{app:preliminaries}

In \cref{sec:preliminaries}, we mentioned that in the white-box model, the agents' valuations $v_i$ are accessed via \emph{polynomial-time algorithms}, which are given explicitly as part of the input. We provide the precise definition of the input model in this section. Formally, we assume that valuations are computed by \emph{Turing machines}. The input to the Turing machine is a list of intervals, and the Turing machine outputs the value of the union of these intervals.

\begin{mdframed}[backgroundcolor=white!90!gray,
      leftmargin=\dimexpr\leftmargin-20pt\relax,
      innerleftmargin=4pt,
      innertopmargin=0pt,
      skipabove=5pt,skipbelow=5pt]
 \begin{definition}[\textup{$\e$-\ch} \textit{(white-box model), formal definition}]
 For any constant $n \geq 1$, and polynomial $p$, the problem \emph{$\e$-\ch\ with $n$ agents} is defined as follows:
 \begin{itemize}
     \item \textbf{Input:} $\e >0$, the Lipschitz parameter $L$, Turing machines $v_1, \dots, v_n$
     \item \textbf{Output:} Any of the following:
     \begin{itemize}
         \item A partition of $[0,1]$ into two sets of intervals $\calI^{+}$ and $\calI^{-}$ such that for each agent $i$, it holds that $|v_i(\calI^{+}) - v_i(\calI^{-})| \leq \e$, using at most $n$ cuts.
         \item A \emph{violation} of $L$-Lipschitz-continuity for one of the valuations.
         \item An input $X$ on which one of the Turing machines does not terminate in at most $p(|X|)$ steps.
         \item (Optional, for monotone valuations only). A \emph{violation} of monotonicity for one of the valuations.
     \end{itemize}
 \end{itemize}
 \end{definition}
 \end{mdframed}\vskip 5pt

Note that for all solution types except the first one, a solution can be a union of an arbitrary number of intervals, i.e., not necessarily obtained using at most $n$ cuts. The violation solutions are only there to ensure that the problem is contained in TFNP. They are irrelevant for our hardness results, which also hold for the promise version of the problem where we are promised that the input does not violate any of the conditions.

\section{The Kuhn Triangulation}\label{sec:kuhn}

Let $D_m := \{0,1/m,2/m, \dots, m/m\}$. Kuhn's triangulation is a standard way to triangulate a grid $D_m^n$. 
Every $x \in (D_m \setminus \{1\})^n$ is the base of the cube containing all vertices $\{y : y_i \in \{x_i,x_i + 1/m\}\}$. Every such cube is subdivided into $n!$ $n$-dimensional simplices as follows: for every permutation $\pi$ of $[n]$, $\sigma = \{y^0,y^1, \dots, y^n\}$ is a simplex, where $y^0 = x$ and $y^i = y^{i-1} + \frac{1}{m}e_{\pi(i)}$ for all $i \in [n]$ (where $e_i$ is the $i$th unit vector).

It is easy to see that Kuhn's triangulation has the following properties:
\begin{itemize}
    \item For any simplex $\sigma = \{z^1, \dots, z^k\}$ it holds that $\|z^i-z^j\|_\infty \leq 1/m$ for all $i,j$, and there exists a permutation $\pi$ of $[k]$ such that $z^{\pi(1)} \leq \dots \leq z^{\pi(k)}$ (component-wise).
    \item Given a point $x \in [0,1]^n$, we can efficiently determine the simplex that contains it as follows. First find the base $y$ of a cube of $D_m^n$ that contains $x$. Next, find a permutation $\pi$ such that $x_{\pi(1)}-y_{\pi(1)} \geq \dots \geq x_{\pi(n)}-y_{\pi(n)}$. Then, it follows that $(y,\pi)$ is the simplex containing $x$.
    \item The triangulation is antipodally anti-symmetric: if $\sigma = \{y^0,y^1, \dots, y^n\}$ is a Kuhn simplex of $D_m^n$, then $\overline{\sigma} = \{\overline{y^0}, \dots, \overline{y^n}\}$ is also a simplex, where $\overline{y^i}_j = 1 - y^i_j$ for all $i,j$.
\end{itemize}
Using Kuhn's triangulation, a function $f: D_m^n \to [-M,M]$ can be extended to a Lipschitz-continuous function $\widehat{f}: [0,1]^n \to [-M,M]$. $\widehat{f}$ is constructed in each Kuhn simplex $\sigma = \{y^0,y^1, \dots, y^n\}$ by interpolating over the values $\{f(y^0),f(y^1), \dots, f(y^n)\}$. In more detail, for any $x \in [0,1]^n$ that lies in simplex $\sigma = \{y^0,y^1, \dots, y^n\}$, we let $z := m \cdot (x-y^0)$. Let $\pi$ denote the permutation used to obtain $\sigma$. Note that since $x$ lies in $\sigma$, it must be that $z_{\pi(1)} \geq z_{\pi(2)} \geq \dots \geq z_{\pi(n)}$. Then, it holds that $x = \sum_{i=0}^n \alpha_i y^i$, where $\alpha_0 = 1 - z_{\pi(1)}$, $\alpha_n = z_{\pi(n)}$, and $\alpha_i = z_{\pi(i)} - z_{\pi(i+1)}$ for $i \in [n-1]$. We define
$$\widehat{f}(x) := \sum_{i=0}^n \alpha_i f(y^i).$$

It is easy to check that the function $\widehat{f}: [0,1]^n \to [-M,M]$ thus obtained is continuous. Indeed, if $x$ lies on a common face of two Kuhn simplices, then the value $\widehat{f}(x)$ obtained by interpolating in either simplex is the same. It can be shown that $\widehat{f}$ is Lipschitz-continuous with Lipschitz parameter $2M \cdot n \cdot m$ with respect to the $\ell_\infty$-norm. Furthermore, if $f$ is antipodally anti-symmetric, i.e., $f(\overline{x}) = - f(x)$ for all $x \in D_m^n$, then so is $\widehat{f}$, i.e., $\widehat{f}(\overline{x}) = - \widehat{f}(x)$ for all $x \in [0,1]^n$.

Finally, if $f$ is monotone, i.e., $f(x) \leq f(y)$ for all $x,y \in D_m^n$ with $x \leq y$, then so is $\widehat{f}$, i.e., $\widehat{f}(x) \leq \widehat{f}(y)$ for all $x,y \in [0,1]^n$ with $x \leq y$. Consider any point $x \in [0,1]^n$ that lies in some simplex $\sigma = \{y^0,y^1, \dots, y^n\}$. Then for any $j \in [n]$ and $t \geq 0$ such that $x+t\cdot e_j$ lies in the simplex $\sigma$, we have
$$\widehat{f}(x+t\cdot e_j) - \widehat{f}(x) = tm f(y^j) - tm f(y^{j-1}) \geq 0$$
since $y^j \geq y^{j-1}$ and $f$ is monotone. Using this it is easy to show that $\widehat{f}$ is monotone within any simplex $\sigma$, since for any $x \leq y$ in $\sigma$ we can construct a path that goes from $x$ to $y$ (and lies in $\sigma$) that only uses the positive cardinal directions. Since monotonicity holds for any segment of the path, it also holds for $x$ and $y$. Finally, for any $x \leq y$ that lie in different simplices, we can just use the straight path that goes from $x$ to $y$, and the fact that $\widehat{f}$ is monotone in each simplex that we traverse.

\end{document}